\documentclass[pra,reprint,letterpaper,nofootinbib,showpacs,floatfix,superscriptaddress]{revtex4-1}
\usepackage{graphicx}
\usepackage{dcolumn}
\usepackage{hyperref}
\usepackage{dcolumn}
\usepackage{amsmath}
\usepackage{hyperref}
\usepackage{amsmath}
\usepackage{amssymb}
\usepackage{amsthm}
\usepackage{algpseudocode}

\newcommand{\be}{\begin{equation}}
\newcommand{\ee}{\end{equation}}
\newcommand{\ba}{\begin{array}}
	\newcommand{\ea}{\end{array}}
\newcommand{\bea}{\begin{eqnarray}}
\newcommand{\eea}{\end{eqnarray}}

\newcommand{\ra}{\rangle}
\newcommand{\la}{\langle}

\newcommand{\calC}{{\cal C }}
\newcommand{\calS}{{\cal S }}

\newcommand{\EE}{\mathbb{E}}

\newtheorem{lemma}{Lemma}

\usepackage{graphicx}
\usepackage{amsmath}
\usepackage{amssymb}
\usepackage{amsthm}
\usepackage{algpseudocode}
\usepackage{algorithmicx}

\usepackage[table,x11names]{xcolor}

\begin{document}

\title{Mitigating measurement errors in multi-qubit experiments }

\author{Sergey Bravyi}
\affiliation{IBM Quantum, T. J. Watson Research Center, Yorktown Heights, NY 10598}
\author{Sarah Sheldon}
\affiliation{IBM Quantum, Almaden Research Center, San Jose, CA 95120}
\author{Abhinav Kandala}
\affiliation{IBM Quantum, T. J. Watson Research Center, Yorktown Heights, NY 10598}
\author{David C. Mckay}
\affiliation{IBM Quantum, T. J. Watson Research Center, Yorktown Heights, NY 10598}
\author{Jay M. Gambetta}
\affiliation{IBM Quantum, T. J. Watson Research Center, Yorktown Heights, NY 10598}

\begin{abstract}
Reducing measurement errors in multi-qubit quantum devices
is critical for performing any quantum algorithm.  
Here we show how to mitigate measurement errors 
by a classical post-processing of the measured outcomes. Our 
techniques apply to  any experiment where measurement outcomes are used for computing
expected values of  observables. 
Two error mitigation schemes are presented based on tensor product and correlated 
Markovian noise models. 
Error rates parameterizing these noise models can be extracted
from the measurement calibration data using a simple formula.
Error mitigation is achieved by applying the inverse noise matrix
to a probability vector that represents the outcomes of a noisy measurement.
The error mitigation overhead, including the the number
of measurements and the cost of the classical post-processing,
is exponential in  $\epsilon n$,
where $\epsilon$ is the maximum error rate and $n$ is the number of qubits. 
We report experimental demonstration of our error mitigation methods
on IBM Quantum devices using stabilizer measurements for graph states
with $n\le 12$ qubits and  entangled $20$-qubit states generated by low-depth random Clifford circuits.
\end{abstract}

\maketitle

\section{introduction}
Quantum computing experiments are beginning a shift from 
few-qubit demonstrations of entangling gates and short quantum circuits to 
larger multi-qubit  quantum algorithms attempting to address 
practically important computational problems.
Although available devices are noisy, error mitigation methods have emerged as a possible 
near-term solution of the fault-tolerance problem~\cite{temme2017error,li2017efficient,li2018practical,otten2018recovering,bonet2018low}.
Error mitigation schemes are particularly attractive
because they introduce no overhead in terms of the number of qubits and gates.
These methods have recently enabled a reliable simulation of shallow quantum circuits
on a noisy hardware without resorting to quantum error correcting codes~\cite{kandala2018extending, havlivcek2019supervised}.
The key idea behind error mitigation is to combine outcomes
of multiple experiments in a way that cancels the contribution of noise to the
quantity of interest, see~\cite{temme2017error,li2017efficient,li2018practical} for details.

Readout errors introduced by imperfect qubit measurements
are often the dominant factor
limiting scalability of near-term devices. 
Here we introduce methods for mitigating readout errors 
for a class of quantum algorithms, where measurement outcomes
are only used for computing mean values of observables.
Notable examples of such algorithms are 
variational quantum eigensolvers~\cite{Peruzzo2014,Kandala2017}, quantum
machine learning \cite{havlivcek2019supervised}, and
tomography of entangled states~\cite{dicarlo2010preparation}.
However, our techniques are not applicable to single-shot
measurements. The latter are indispensable to many 
applications such as quantum teleportation,
measurement-based computation, and quantum error correction.

Prior work~\cite{chen2019detector, chow2010detecting, ryan2015tomography, merkel2013self, chow2012universal} demonstrated that readout errors in quantum devices
based on superconducting qubits
can be well understood in terms of purely classical noise models.
Such models describe a noisy $n$-qubit measurement
 by a matrix of transition probabilities $A$
 of size $2^n\times 2^n$
such that $A_{y,x}$ is  the probability of observing
a measurement outcome $y$ provided that the true outcome is $x$.
It is common to simplify the noise model further by assuming that the
noise acts independently on each qubit~\cite{sun2018efficient, Kandala2017}.
This defines a tensor product noise model such that 
$A$ is a tensor product of $2\times 2$ noise matrices.
The model depends on $2n$ error rates describing single-qubit
readout 
errors $0\to 1$ and $1\to 0$. The open-source software package, Qiskit~\cite{qiskit}, 
provides a toolkit for calibrating $A$ and mitigating readout errors
for either a tensor product or a general noise model. 

Although the tensor product noise model is appealingly simple, it leaves aside cross-talk
errors encountered in real-world setups.   Cross-talk during readout can arise from the underlying qubit-qubit coupling, 
and spectral overlap of readout resonators with stray couplings or multiplexing \cite{heinsoo}.
Here we introduce a correlated noise model
based on Continuous Time Markov Processes (CTMP).
The corresponding noise matrix has the form $A=e^G$, where
$G$ is a sum of local operators generating single- and two-qubit
readout errors such as $01\to 10$, $11\to 00$, etc.
The model depends on $2n^2$ error rates.

Measurement calibration aims at learning the 
unknown parameters of a noise model.
This is achieved by preparing a set of
well-characterized input states (we use the standard basis vectors) 
and repeatedly performing $n$-qubit measurements on each input state.
We show how to extract 
parameters of the considered noise models 
from the  calibration data
using simple analytic formulas. 
The number of input states required to calibrate
the tensor product and CTMP noise models scales
linearly with the number of qubits $n$.

Once parameters of the noise model are known, 
error mitigation proceeds by applying the inverse  noise
matrix to a probability vector that represents the noisy measurement outcomes.
We show how to sidestep the explicit computation of the inverse noise matrix,
which may be prohibitive for a large number of qubits.
Instead, we follow ideas of~\cite{temme2017error} and 
cancel errors by sampling noise matrices from a suitable quasi-probability
distribution. The total error mitigation overhead,
including the number of measurements and the cost of classical post-processing,
scales as $e^{4\gamma} poly(n)$, where
$\gamma$ is a parameter quantifying the noise strength,
as formally defined below.
For example, if each qubit independently undergoes a bit-flip readout error
with probability $\epsilon$ then $\gamma = n\epsilon$.
Thus our methods are practical
whenever $\gamma$ is a small constant.  
For comparison, all previously studied
readout error mitigation techniques require manipulations with 
probability vectors of size $2^n$ which limits their scalability.
Importantly, both tensor product and CTMP error mitigation methods
provide an unbiased estimator of the ideal mean value to be measured.
For a fixed number of qubits, the cost of approximating the ideal mean
value with a specified error tolerance $\delta$ scales as $\delta^{-2}$.

The proposed  error mitigation techniques are demonstrated using 
a $20$-qubit
IBM Quantum device. The device has a few-percent readout error rate.
For small number of qubits ($n\le 7$), a comparison is performed
between the full noise matrix $A$, the tensor product,
and the CTMP noise models. It is found that the CTMP model provides
a considerably more accurate approximation of the readout noise.
Error mitigated stabilizer measurements  of $n$-qubit graph states
are reported for $n\le 12$. Finally, we perform error-mitigated
stabilizer measurements on $20$-qubit states generated
by low-depth random Clifford circuits. 
The classical post-processing
associated with the error mitigation, including the measurement
calibration and the noise inversion steps, can be done on a laptop
in a few-second time frame
for each of these experiments.
Note that the error mitigation runtime is crucial for 
VQE-type algorithms~\cite{Kandala2017} 
where quantum mean value estimation is repeatedly used
in a feedback loop.

Let us briefly discuss the previous work.
Seif et al~\cite{seif2018machine} considered
a system of trapped-ion qubits  and the problem of 
discriminating between different basis states based
on data collected by an array of photon detectors.
It was shown that a neural network trained on the measurement
calibration data can realize a high-fidelity state discrimination~\cite{seif2018machine}.
It remains to be seen whether similar methods can address
the problem of 
high-precision mean value estimation considered here.
Our techniques build off earlier work from Kandala et al~\cite{Kandala2017} and are most closely related to the work
of Sun and Geller~\cite{sun2018efficient,geller2020rigorous,geller2020efficient} who showed
how to characterize readout errors
in multi-qubit devices. In particular, Ref.~\cite{sun2018efficient}
introduced a noise model based on single- and two-qubit correlation
functions which can be viewed as an analogue of the CTMP
model considered here.
Ref.~\cite{maciejewski2019mitigation} elucidated 
how performance of readout error mitigation methods
is affected
by the quality of the measurement calibration, number of measurement samples, 
and the presence of coherent (non-classical) measurement errors,
see also~\cite{chen2019detector}.
Ref.~\cite{nachman2019unfolding} pointed out a 
surprising connection between quantum readout error mitigation
methods and unfolding algorithms used in high energy physics
to cope with a finite resolution of particle detectors~\cite{cowan2002survey}.
However the classical processing cost of unfolding methods may be prohibitive
for multi-qubit experiments  since they require  manipulations with 
exponentially large  probability vectors.
Applications of the noise matrix inversion method
in the context of variational quantum algorithms are discussed
in~\cite{hamilton2019error}. 

After completion of the present paper we became aware of a closely related work
by Hamilton, Kharazi et al.~\cite{hamilton2020scalable} which developed a scalable characterization of 
readout errors based on a cumulant expansion. The latter can be viewed
as a counterpart of the CTMP noise model considered here. 
Namely,  both methods provide a recipe for combining single- and two-qubit noise matrices into 
the full noise matrix describing a multi-qubit device.

The paper is organized as follows. 
Section~\ref{sec:unbiased} defines an error-mitigated
mean value of an observable
and shows that it provides an unbiased estimator of the ideal
mean value. The tensor product and CTMP
noise models are discussed in Sections~\ref{sec:product},\ref{sec:correlated}.
The parameters of these noise models can be extracted from the measurement
calibration data as described in Section~\ref{sec:calib}.
The experimental demonstration of our error mitigation methods
is reported in Section~\ref{sec:experiment}.
Appendices~A,B,C prove technical lemmas used in the main text and 
provide additional details on the implementation of our algorithms.
The experimental hardware is described in Appendix~D.

\section{Unbiased error mitigation}
\label{sec:unbiased}

Let $\rho$ be the output state of a (noisy) quantum circuit
acting on $n$ qubits.
Our goal is to measure the expected value of a given 
observable $O$ on the state $\rho$ within a specified precision $\delta$.
For simplicity, below we assume that the observable $O$ is diagonal in the
standard basis and takes values in the range $[-1,1]$, that is,
\[
O=\sum_{x\in \{0,1\}^n} O(x) |x\ra\la x|, \qquad |O(x)|\le 1.
\]
Consider $M$ independent experiments where each
experiment prepares the state $\rho$ and then measures
each qubit in the standard basis. Let $s^i\in \{0,1\}^n$
be the string of measurement outcomes
observed in the $i$-th experiment.
In the absence of measurement errors
one can approximate the mean value $\mathrm{Tr}(\rho O)$
by an empirical mean value 
\be
\label{xi_ideal}
\mu = M^{-1} \sum_{i=1}^M O(s^i).
\ee
By Hoeffding's inequality,  
$|\mu - \mathrm{Tr}(\rho O)|\le \delta$ with high probability
(at least $2/3$) if we choose $M= 4\delta^{-2}$. Furthermore,
$\mu$ is an unbiased estimator of $\mathrm{Tr}(\rho O)$.

Suppose now that measurements are noisy. 
The most general model of a noisy $n$-qubit measurement
involves a POVM with $2^n$ elements  $\{\Pi_x\}$
labeled by $n$-bit strings $x$. The probability of observing
a measurement outcome $x$ is given by  $P(x)=\mathrm{Tr}(\rho \Pi_x)$.
In the ideal case one has $\Pi_x=|x\ra\la x|$ while
in the presence of noise $\Pi_x$ could be arbitrary positive semi-definite
operators that obey the normalization condition $\sum_x \Pi_x = I$.
To enable efficient error mitigation, we make a simplifying assumption
that all POVM elements $\Pi_x$ are diagonal in the standard basis,
that is,
\[
\Pi_y = \sum_x \la y|A|x\ra \, |x\ra\la x|
\]
for some stochastic matrix $A$ of size $2^n$.
In other words, $\la y|A|x\ra$ is
the probability of observing an outcome $y$ provided that the true outcome is $x$.
For example, suppose $n=1$ and 
\[
A = \left[ \ba{cc} 0.9 & 0.2 \\ 0.1 & 0.8 \\ \ea \right].
\]
Then the probability of measuring $1$ on a qubit prepared in the state $|0\ra$
is $10\%$. The probability of measuring $0$ on a qubit prepared in the state $|1\ra$
is $20\%$. For brevity, we shall refer to such readout errors
as $0\to 1$ and $1\to 0$ respectively.
In this section we assume
that the matrix $A$ is known (this assumption is relaxed in Section~\ref{sec:calib}).
By analogy with Eq.~(\ref{xi_ideal}),
define an error-mitigated empirical mean value as
\be
\label{xi}
\xi = M^{-1} \sum_{i=1}^M \; \sum_{x} O(x) \la x|A^{-1} |s^i\ra.
\ee
\begin{lemma}
\label{lemma:xi}
The random variable $\xi$  is
an unbiased estimator of $\mathrm{Tr}(\rho O)$ with
the standard deviation $\sigma_\xi \le \Gamma M^{-1/2}$,
where 
\be
\label{Gamma}
\Gamma =  \max_y \sum_x  |\la x|A^{-1} |y\ra|.
\ee
\end{lemma}
A proof of the lemma is given in Appendix~\ref{app:1}.
By Hoeffding's inequality,  
$|\xi - \mathrm{Tr}(\rho O)|\le \delta$ with high probability
(at least $2/3$)  if the number of measurements is 
\be
\label{shots_upper_bound}
M = 4\delta^{-2} \Gamma^2.
\ee
Recall that $M\sim\delta^{-2}$ in the ideal case.
Thus the quantity $\Gamma^2$ can be viewed as an error mitigation overhead.
In other words, error mitigation increases
the number of measurements required to achieve a given precision $\delta$  by 
a factor of $\Gamma^2$ compared with the case of ideal measurements.

\section{Tensor product noise}
\label{sec:product}

 Consider first a simple case
when  $A$ is a tensor product of $2\times 2$ stochastic matrices,
\be
\label{product_noise}
A = \left[ \ba{cc} 1-\epsilon_1 & \eta_1 \\
\epsilon_1 & 1-\eta_1 \\ \ea \right]
\otimes \cdots \otimes 
\left[ \ba{cc} 1-\epsilon_n & \eta_n \\
\epsilon_n & 1-\eta_n \\ \ea \right].
\ee
Here $\epsilon_j$ and $\eta_j$ are error rates
describing readout errors $0\to 1$ and $1\to 0$ respectively.
The error-mitigated mean value $\xi$ defined in Eq.~(\ref{xi}) can be easily computed 
when the observable $O$ has a tensor product form,
\[
O=O_1\otimes \cdots \otimes O_n.
\]
For example, $O$ could be a product of Pauli $Z$
operators on some subset of qubits.
Alternatively, $O$ could project some subset of qubits
onto $0$ or $1$ states.
Define a single-qubit state $|e\ra=|0\ra+|1\ra$.
From Eqs.~(\ref{xi},\ref{product_noise}) one  gets
\be
\label{productO}
\xi = M^{-1} \sum_{i=1}^M \;  \prod_{j=1}^n  \la e|O_j \left[ \ba{cc} 1-\epsilon_j & \eta_j \\
\epsilon_j & 1-\eta_j \\ \ea \right]^{-1} |s^i_j\ra.
\ee
Here $s^i \in \{0,1\}^n$ is the outcome observed in the $i$-th measurement
and $s^i_j\in \{0,1\}$ is $j$-th bit of $s^i$. 
The single-qubit matrix element that appears in Eq.~(\ref{productO}) can be computed
in constant time for any fixed pair $i,j$. Thus one can compute $\xi$  in time $\approx nM$,
that is, the classical post-processing runtime is linear in the
number of qubits and the number of measurements.
The above algorithm is generalized to arbitrary observables $O$
in Appendix~\ref{app:product}.
The algorithm has runtime 
 $\approx nM\tau$, where $\tau$
 is the cost of computing the observable $O(x)$
 for a given $x$.

Suppose the noise is weak such that
$\epsilon_j,\eta_j\le 1/2$ for all $j$.
Substituting the definition of $A$ into Eq.~(\ref{Gamma})
one gets
\be
\label{GammaProduct}
\Gamma= \prod_{j=1}^n \frac{1+|\epsilon_j-\eta_j|}{1-\epsilon_j-\eta_j}.
\ee
We show how to extract the error rates $\epsilon_j$ and $\eta_j$
from the readout calibration data in Section~\ref{sec:calib}.
Then Eqs.~(\ref{shots_upper_bound},\ref{GammaProduct})
determine the number of measurements $M$, as a function of the desired precision $\delta$.
Assuming that $\epsilon_j,\eta_j\ll 1$ one gets
$\Gamma\approx e^{2 \gamma}$, where $\gamma$ is the {\em noise strength}
defined as
\[
\gamma = \sum_{j=1}^n \max{\{ \epsilon_j,\eta_j\}}.
\]
The error mitigation overhead thus scales as $\Gamma^2\approx e^{4\gamma}$.

\section{Correlated Markovian noise}
\label{sec:correlated}

How can one extend the tensor product noise model to account for
correlated (cross-talk) errors ? A natural generalization of a tensor product 
operator  is a Matrix Product Operator (MPO)~\cite{schollwock2011density}.
Unfortunately, we empirically observed that MPO-based noise models 
poorly agree with the experimental data, unless the MPO bond dimension is unreasonably large.
Noise models based on neural networks provide another possible extension~\cite{seif2018machine}.
However, training such noise models and computing error-mitigated mean values in a few-second
time frame might not be possible for multi-qubit experiments.
Instead, here we propose noise models based on 
Continuous Time Markov Processes (CTMP)~\cite{gillespie1977exact}.

Define
a CTMP noise model as a stochastic matrix $A$ of size $2^n\times 2^n$ such 
that~\footnote{Although our definition makes no reference to time,
one can express $A$ as a solution of a differential equation
$\dot{A}(t)=GA(t)$ with the initial condition $A(0)=I$. Thus $G$
serves as a generator of a continuous time Markov process
such that $\la y|G|x\ra$ is the rate of transitions from a state $x$
to a state $y\ne x$.}
\be
\label{CTMP}
A=e^G, \qquad G= \sum_{i=1}^{2n^2} r_i G_i
\ee
where $e^G=\sum_{p=0}^\infty G^p/p!$ is the matrix exponential,
$r_i\ge 0$ are error rates, and $G_i$ 
are single-qubit or two-qubit operators
from the following list

\begin{center}
\begin{tabular}{c|c|c}
\toprule
Generator $G_i$ & Readout error & Number of generators\\
\hline
$|1\ra\la 0| - |0\ra\la 0|$ & $0\to 1$ & $n$ \\
\hline
$|0\ra\la 1| - |1\ra\la 1|$ & $1\to 0$ & $n$ \\
\hline
$|10\ra\la 01| - |01\ra\la 01|$ & $01\to 10$ & $n(n-1)$ \\
\hline
$|11\ra\la 00|-|00\ra\la 00|$ & $00\to 11$ & $n(n-1)/2$ \\
\hline
$|00\ra\la 11| - |11\ra\la 11|$ & $11\to 00$ & $n(n-1)/2$ \\
\botrule
\end{tabular}
\end{center}
Each operator $G_i$ generates a readout error on
some bit or some pair of bits. 
The right column shows the number of ways to choose
qubit(s) acted upon by a generator.
The negative terms in $G_i$ 
ensure that $A$ is a stochastic matrix.
The  CTMP model depends on $2n^2$ parameters $r_i$, as can be seen
by counting the number of generators $G_i$ of each type.
We show how to extract the error rates $r_i$
from the readout calibration data  in Section~\ref{sec:calib}.
Furthermore, 
the tensor product model Eq.~(\ref{product_noise})
is a special case of  CTMP with $r_i=0$ for
all two-qubit errors. Indeed, in this case Eq.~(\ref{CTMP})
defines a tensor product of single-qubit stochastic matrices
\[
\left[ \ba{cc} 1-\epsilon & \eta \\
\epsilon & 1-\eta \\ \ea \right]=
e^G, \quad G = -\frac{\log{(1-\epsilon-\eta)} }{\epsilon+\eta} 
\left[ \ba{cc} -\epsilon & \eta \\
\epsilon & -\eta \\ \ea \right].
\]
Here we assume that $\epsilon+\eta<1$.

Next let us show how to compute  the error-mitigated 
mean value $\xi$ defined in Eq.~(\ref{xi}).
Since $\xi$ itself is only a $\delta$-estimate of the true 
mean value $\mathrm{Tr}(\rho O)$, 
it suffices to approximate  $\xi$ within an error 
$\approx \delta$.  
We use a version of the well-known  Gillespie algorithm~\cite{gillespie1977exact}
for simulating continuous Markov processes 
combined with the following lemma proved in Appendix~\ref{app:stochastic}.
\begin{lemma}
	\label{lemma:QPR}
	Suppose $A$ is an invertible stochastic matrix. 
	There exist  stochastic matrices
	$S_\alpha$ and real coefficients $c_\alpha$ such that $\|c\|_1 \equiv \sum_\alpha |c_\alpha|\ge  \Gamma$ and
	\be
	\label{QPR}
	A^{-1}=\sum_\alpha c_\alpha S_\alpha.
	\ee
	Furthermore, $\|c\|_1= \Gamma$ for some decomposition as above.
\end{lemma}
Here $\Gamma$ is the quantity defined in Eq.~(\ref{Gamma}).
Suppose we have found a decomposition Eq.~(\ref{QPR}) for the 
inverse CTMP
noise matrix $A^{-1}=e^{-G}$.
Then the error-mitigated mean value  $\xi$ defined in Eq.~(\ref{xi})
can be approximated using the following algorithm. 
\begin{center}
	\fbox{\parbox{\linewidth}{
			{\bf Algorithm 1: }
			\begin{algorithmic}
				\State{$T\gets  4\delta^{-2}\|c\|_1^2$}
				\For{$t=1$ to $T$}
				\State{Sample $i \in [M]$ uniformly at random}
				\State{Sample $\alpha$ from the distribution $q_\alpha=|c_\alpha|/\|c\|_1$}
				\State{Sample $x$ from the distribution $\la x|S_\alpha |s^i\ra$}
				\State{$\xi_t\gets \mathrm{sgn}(c_\alpha)O(x)$}
				\EndFor
				\State{\Return{$\xi'=T^{-1}\|c\|_1\sum_{t=1}^T \xi_t$}}
			\end{algorithmic}
	}}
\end{center}
Here $c_\alpha,S_\alpha$ are the same as in Lemma~\ref{lemma:QPR}.
We claim that the output 
 this algorithm satisfies $|\xi' -\xi|\le \delta$  with high probability
(at least $2/3$). 
Indeed,  substituting Eq.~(\ref{QPR}) into Eq.~(\ref{xi}) one gets
\be
\label{QPR1}
\xi = M^{-1} \|c\|_1 \sum_{i=1}^M 
\sum_\alpha \sum_x  q_\alpha \mathrm{sgn}(c_\alpha) O(x) \la x|S_\alpha |s^i\ra.
\ee
that is,
\be
\label{QPR2}
\xi = \|c\|_1 \EE_{i,\alpha,x}  \, \mathrm{sgn}(c_\alpha)O(x(i,\alpha)).
\ee
Here $i\in [M]$ is picked uniformly at random,
$\alpha$ is sampled from  $q_\alpha$,
and $x(i,\alpha)$ is sampled from the distribution $\la x|S_\alpha|s^i\ra$
with fixed $i,\alpha$.
Accordingly, a random variable
$\|c\|_1\mathrm{sgn}(c_\alpha)O(x(i,\alpha))$ is an unbiased estimator of $\xi$
with the variance at most $\|c\|_1^2$.
By Hoeffding's inequality,  
$|\xi' - \xi|\le \delta$ with high probability
(at least $2/3$).


In the special case of the tensor product noise
one can efficiently compute a decomposition 
$A^{-1}=\sum_\alpha c_\alpha S_\alpha$
defined in Lemma~\ref{lemma:QPR}
with $\|c\|_1=\Gamma$,
see Appendix~\ref{app:product}.
This yields an
instantiation of Algorithm~1 that computes a $\delta$-estimate of $\xi$ in time roughly
\be
\label{runtime1}
n\Gamma^2 \delta^{-2} \sim nM.
\ee
Here $M$ is the number of measurements determined by Eq.~(\ref{shots_upper_bound}).
Further details on the implementation of Algorithm~1
for the tensor product noise
can be found in Appendix~\ref{app:product}.

Next, let us show how to implement Algorithm~1 for
the CTMP model $A=e^G$.
Define a parameter
\be
\label{gammaG}
\gamma = \max_{x\in \{0,1\}^n} -\la x|G|x\ra.
\ee
We shall see that the error mitigation overhead introduced by the CTMP
method scales as $e^{4\gamma}$. To avoid a confusion with individual error rates
$r_i$ we shall refer to the quantity $\gamma$ as the CTMP noise strength.
The function $-\la x|G|x\ra$ can be viewed as a classical Ising-like Hamiltonian
with two-spin interactions. Although finding the maximum energy
of such Hamiltonians is NP-hard in the worst case, this can be easily done
for moderate system sizes, say $n\le 50$, using heuristic optimizers
such as simulated annealing. Below we assume that the noise strength $\gamma$
has been already computed.

First we note  that $\gamma\ge 0$ since  $G$ has non-positive diagonal elements.
Define a matrix
\be
\label{Bmatrix}
B= I + \gamma^{-1} G.
\ee
Using Eq.~(\ref{gammaG}) and the fact that $G$
has zero column sums,
one can check
that $B$ is a stochastic matrix. Furthermore,
\be
\label{QPR3}
A^{-1} = e^{-G} = e^\gamma \cdot e^{-\gamma B}
=\sum_{\alpha=0}^\infty \frac{e^\gamma  (-\gamma)^\alpha}{\alpha!} B^\alpha
\equiv
\sum_{\alpha=0}^\infty c_\alpha  S_\alpha
\ee
where
\be
\label{QPR3'}
S_\alpha = B^\alpha \quad \mbox{and} \quad c_\alpha =  \frac{e^\gamma  (-\gamma)^\alpha}{\alpha!}.
\ee
Clearly, $S_\alpha$ is a stochastic matrix for any integer $\alpha\ge 0$.
Thus Eqs.~(\ref{QPR3},\ref{QPR3'}) provide a stochastic 
decomposition of $A^{-1}$ stated in Lemma~\ref{lemma:QPR}
with the $1$-norm
\be
\label{QPR4}
\|c\|_1 = \sum_{\alpha=0}^\infty  \frac{e^\gamma \gamma^\alpha}{\alpha!} = e^{2\gamma}.
\ee
The probability distribution 
\be
\label{QPR5}
q_\alpha = |c_\alpha|/\|c\|_1 = \frac{e^{-\gamma} \gamma^\alpha}{\alpha!}
\ee
is the Poisson distribution with the mean $\gamma$.
Substitute the decomposition
of Eqs.~(\ref{QPR3},\ref{QPR3'}) into Algorithm~1.
The only non-trivial step of the algorithm
is sampling $x$ from the distribution 
\[
\la x|S_\alpha|s^i\ra = \la x|B^\alpha|s^i\ra.
\]
This amounts to simulating $\alpha$ steps 
of a Markov Chain with the transition matrix $B$. 
Note that $B$ is a sparse efficiently computable matrix.
More precisely, each column of $B$ has $\approx n^2$ non-zero elements.
Thus one can simulate a single step of the Markov chain in time $\approx n^2$.
Since $\EE(\alpha)=\gamma$, 
simulating $\alpha$ steps of the Markov chain on average
takes time $\approx \gamma n^2$.
The overall runtime of Algorithm~1 becomes
\[
n^2 \gamma \|c\|_1^2 \delta^{-2} \sim  n^2 \gamma e^{4\gamma} \delta^{-2}.
\]
Here we used Eq.~(\ref{QPR4}). Assuming that cross-talk errors
result from short-range interactions between qubits, one should expect
that each qubit participates in a constant number of  generators $G_i$
independent of $n$.
If this is the case, each column of $G$ has at most $Cn$ non-zero elements
for some constant $C$. Then $\gamma \le C\epsilon n$, 
where $\epsilon = \max_i r_i$
is the maximum error rate. Thus
the exponential term in the runtime of 
Algorithm~1 becomes $e^{4 C\epsilon n}$.
Experimental data presented in 
Section~\ref{sec:experiment} are consistent with 
a linear scaling $\gamma \sim n$, see Fig.~\ref{fig:gamma}.

\section{Measurement calibration}
\label{sec:calib}

Calibration aims at learning
parameters of the noise model from the experimental data.
A single calibration round initializes the $n$-qubit register
in a basis state $|x\ra$  and performs a noisy measurement
of each qubit, keeping the record of the measured outcome $y$.
Let $m(y,x)$ be the number of rounds with the input state $x$ and the measured
output state $y$.
We fix some set of input states $\calC  \subseteq \{0,1\}^n$
and perform $N_{\mathsf{cal}}$ calibration rounds for each input state
$x\in \calC$. Thus $\sum_y m(y,x)=N_{\mathsf{cal}}$ if $x\in \calC$
and $m(y,x)=0$ if $x\notin \calC$.
The calibration requires $N_{\mathsf{cal}} \cdot |\calC|$ experiments in total.

Learning the full noise matrix 
requires calibration of each possible
input state, i.e.  $\calC= \{0,1\}^n$. Let
$A_{\mathsf{full}}$ be the empirical estimate of the
full noise matrix, i.e. 
 \be
 \label{cal_full}
 \la y | A_{\mathsf{full}}|x\ra  =\frac{m(y,x)}{N_{\mathsf{cal}}}
 \ee 
where $x,y\in \{0,1\}^n$.

To learn error rates of the tensor product or the CTMP model
it suffices to calibrate a small subset of basis states.
Consider first the tensor product model.
To ensure that each possible single-qubit readout error
is probed on some input state, we
assume that for any qubit $j$ 
the set $\calC$ contains at least one state $x$ with $x_j=0$
and at least one state with $x_j=1$.
Let $\epsilon_j$ and $\eta_j$  be the rates
of errors $0\to 1$ and $1\to 0$ on the $j$-th qubit, see Eq.~(\ref{product_noise}).
We set
\begin{align}
\epsilon_j =  & \frac{\sum_{x,y} m(y,x) \la 1|y_j\ra \la x_j|0\ra}{\sum_{x,y} m(y,x)\la x_j|0\ra}, \nonumber\\
\eta_j = & \frac{\sum_{x,y} m(y,x) \la 0|y_j\ra \la x_j|1\ra}{\sum_{x,y} m(y,x)\la x_j|1\ra}\label{calTP}
\end{align}
For example, $\epsilon_j$ is the fraction of calibration rounds
that resulted in a readout error $0\to 1$  on the $j$-th qubit
among all rounds with $x_j=0$.
Let $A_{\mathsf{tp}}$ be the tensor product noise matrix
defined
by the rates $\epsilon_j, \eta_j$ and Eq.~(\ref{product_noise}).

Next consider the CTMP model. Recall that the model is defined by $2n^2$ error rates 
describing readout errors $0\leftrightarrow 1$, $01\leftrightarrow 10$, and $00\leftrightarrow 11$.
To ensure that each possible two-qubit readout error
is probed on some input state, we
assume that for any pair of qubits $j\ne k$
and for any $v_j,v_k\in  \{0,1\}$ the set $\calC$ contains at least one state
$x$ such that $x_j=v_j$ and $x_k=v_k$.
For example, $\calC$ could include bit strings $0^n$, $1^n$, and all weight-$1$ strings. 
Fix a pair of qubits  $j\ne k$. Given 
a bit string $x\in \{0,1\}^n$, let $x_{\mathsf{in}}\in \{0,1\}^2$ be the restriction of $x$ onto the bits $j,k$
and $x_{\mathsf{out}}\in \{0,1\}^{n-2}$ be the restriction of $x$ onto all remaining bits
$i\notin \{j,k\}$. 
We choose  CTMP error rates on qubits $j,k$
by defining a local noise matrix $A(j,k)$ 
and fitting $A(j,k)$ with a two-qubit CTMP model. 
More formally, 
$A(j,k)$ is a stochastic matrix of size $4\times 4$ with matrix elements
\[
\la w| A(j,k)|v\ra =  
 \frac{\sum_{x,y} m(y,x) \la w|y_{\mathsf{in}}\ra \la x_{\mathsf{in}}|v\ra \la x_{\mathsf{out}}|y_{\mathsf{out}}\ra }{\sum_{x,y} m(y,x)\la x_{\mathsf{in}}|v\ra \la x_{\mathsf{out}}|y_{\mathsf{out}}\ra}.
\]
Here $v,w\in \{0,1\}^2$.
In words, $\la w| A(j,k)|v\ra$ is the fraction of calibration rounds 
that resulted in a readout error $v\to w$ on qubits $j,k$ 
and no errors on qubits $i\notin \{j,k\}$
among all rounds with $x_{\mathsf{in}}=v$ that resulted in no errors 
on qubits $i\notin \{j,k\}$. Let
\[
G(j,k) = \log{A(j,k)}.
\]
be the matrix logarithm of $A(j,k)$.
We choose the branch of the $\log$ function such that
$G(j,k)=0$ if $A(j,k)=I$. 
The rates of two-qubit errors occurring on the qubits $j,k$ are
chosen as shown in the following table.

\begin{center}
\begin{tabular}{c|c}
\toprule
Readout error  & CTMP error rate \\
\hline
$01\to 10$ & $\la 10 |G'(j,k)|01\ra$  \\
\hline
$10\to 01$ & $\la 01 |G'(j,k)|10\ra$ \\
\hline
$00\to 11$ &  $\la 11 |G'(j,k)|00\ra$ \\
\hline
$11\to 00$ & $\la 00|G'(j,k)|11\ra$  \\
\botrule
\end{tabular}
\end{center}
Here $G'(j,k)$ is a matrix obtained from $G(j,k)$ by setting
to zero all negative off-diagonal elements. 
Let $r_j^{0\to 1}$ and $r_j^{1\to 0}$ be the rates
of single-qubit readout errors $0\to 1$ and $1\to 0$ occurring on the $j$-th qubit.
We set
\[
r_j^{0\to 1} = \frac1{2(n-1)}  \sum_{k\ne j}  \la 10|G'(j,k)|00\ra + \la 11|G'(j,k)|01\ra,
\]
\[
r_j^{1\to 0} = \frac1{2(n-1)}  \sum_{k\ne j}  \la 00|G'(j,k)|10\ra + \la 01|G'(j,k)|11\ra.
\]
Here we  noted that two-qubit errors
$01\leftrightarrow 10$ or $00\leftrightarrow 11$ on qubits $j,k$
can only be generated by the local noise matrix $A(j,k)$.
Meanwhile, single-qubit errors $0\leftrightarrow 1$ on the $j$-th qubit
can be generated by any local noise matrix $A(j,k)$ with $k\ne j$.
The above definition of single-qubit error rates amounts to
taking the average over all such possibilities.
Let $A_{\mathsf{ctmp}}$ be the CTMP noise model defined by 
the above error rates and Eq.~(\ref{CTMP}).

The calibration method described above is well-defined only if
the set of input states $\calC$ is {\em complete} in the sense that
any two-bit readout error can be probed on some input state
$x\in \calC$. More formally, let us say that a set of $n$-bit strings $\calC$
is complete if  for any pair of bits $1\le a<b\le n$ and for any 
bit values $\alpha,\beta\in \{0,1\}$ there exists at least one string
$x\in \calC$ such that $x_a=\alpha$ and $x_b=\beta$,
see Fig.~\ref{fig:cal_input} for an example. 
Below we describe complete sets of input states used
in our experiments.

\begin{figure}[h]
	\includegraphics[height=2.5cm]{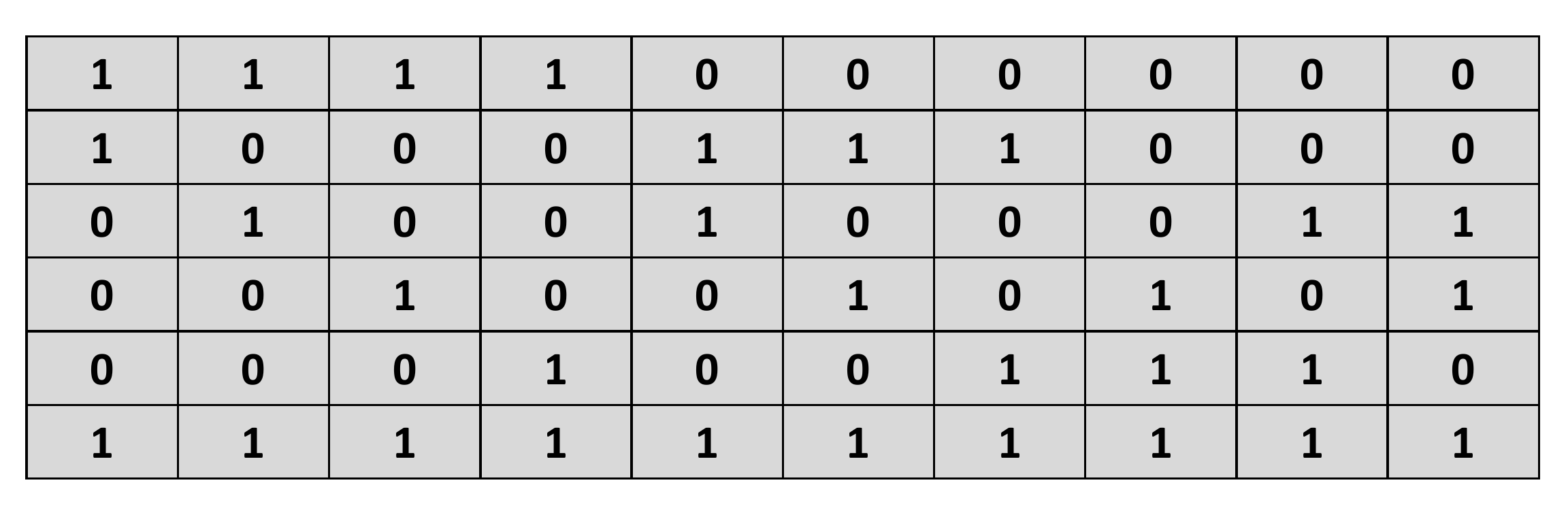}
	\caption{
	A complete set of input states for CTMP calibration on $n=10$ qubits.
        Each row of the matrix defines a bit string $x\in \calC$. 
	One can check that any pair of columns contains $00$, $01$, $10$, and $11$
in at least one row. Thus any two-bit readout error can be probed on some input $x\in \calC$.
We numerically checked that $|\calC|\ge 6$ for any complete set $\calC\subseteq \{0,1\}^{10}$.
 \label{fig:cal_input}}
\end{figure}

\noindent {\em Weight-$1$ calibration:}
the set $\calC$ includes all $n$-bit strings with the Hamming weight
$0$, $1$, $n$. For example, if $n=4$ then $\calC$
includes $0000$,  $1111$, and all permutations of $1000$.\\

\noindent {\em Weight-$2$ calibration:}
the set $\calC$ includes all $n$-bit strings with the Hamming 
weight $0,1,2$. For example, if $n=4$ then
$\calC$ includes $0000$,  all permutations of $1000$,
and all permutations of $1100$.\\

\noindent {\em Hadamard calibration:}  
Let $p$ be the smallest integer such that $n<2^p$.
Given an integer $a$ in the range $[0,2^p-1]$, let
$a_1,\ldots,a_p\in \{0,1\}$ be the binary digits of $a$.
We choose $\calC=\{x^0,\ldots,x^{2^p-1} \}$, where
$x^a\in \{0,1\}^n$ is defined by
\[
x^a_b = \sum_{i=1}^p a_i b_i {\pmod 2}
\]
for $b=1,\ldots,n$. For example, if $n=4$ then $\calC$
includes all even-weight bit strings $x\in \{0,1\}^4$.
In general, $|\calC|\le 2n$. One can easily check that each
two-bit error is probed on exactly $2^{p-2}$ input states
$x\in \calC$. Thus the calibration resources are allocated 
evenly among all possible errors.

\section{Experimental results}
\label{sec:experiment}

All experiments are performed using
$n$-qubit registers  of the 20-qubit IBM Quantum device
called $\mathsf{ibmq\_johannesburg}$,
see Appendix~\ref{app:hardware} for details.
Only the chosen $n$-qubit register is used for readout calibration.
The unused qubits are initialized in the $|0\ra$ state in each calibration round.
Measurement outcomes on the unused qubits are ignored.
All experiments were performed with $N_{\mathsf{cal}}=8192$ calibration rounds.

First, let us discuss how well the tensor product (TP) and the CTMP models
approximate the readout noise observed in the experiment. 
We  quantify the difference between a pair of $n$-qubit noise matrices
$A$ and $A'$ using the Total Variation Distance (TVD)
\[
\mathsf{TVD}(A,A') =\frac12 \; \max_{x} \sum_y | \la y|A|x\ra - \la y|A'|x\ra|
\]
The TVD provides an upper bound on the probability
of distinguishing output distributions $A|p\ra$ and $A'|p\ra$ for
any input distribution $p$. This enables a comparison between noise
models in the worst-case scenario rather than for a specific quantum state
and a specific observable.

For small number of qubits, $n\le 7$,
all three calibration methods 
described in Section~\ref{sec:calib}
have been performed. Let
$A_{\mathsf{full}}$ be the full noise matrix.
Let  $A_{\mathsf{tp}}$  and $A_{\mathsf{ctmp}}$
be the noise matrices predicted by 
the TP and CTMP models with weight-$2$ calibration.
Fig.~\ref{fig:TVDplot}  shows the observed distances 
$\mathsf{TVD}(A_{\mathsf{full}},A_{\mathsf{ctmp}})$
and  $\mathsf{TVD}(A_{\mathsf{full}},A_{\mathsf{tp}})$
for $n=4,5,6,7$ qubits (with the qubits chosen for $n<7$ consisting of subsets of the $n=7$ case) and $16$ independent experiments
per each number of qubits.
It indicates that the CTMP model provides a more accurate
approximation of the readout noise achieving
2X reduction in the TVD metric for $n=6,7$ qubits.  The increased separation between TP and CTMP models for larger numbers of qubits indicates that readout crosstalk is larger in the additional qubits used for the circuits with more qubits.
Note that this improvement comes at no additional cost
since both models have access to exactly the same calibration data set.

\begin{figure}[h]
	\includegraphics[width=\columnwidth]{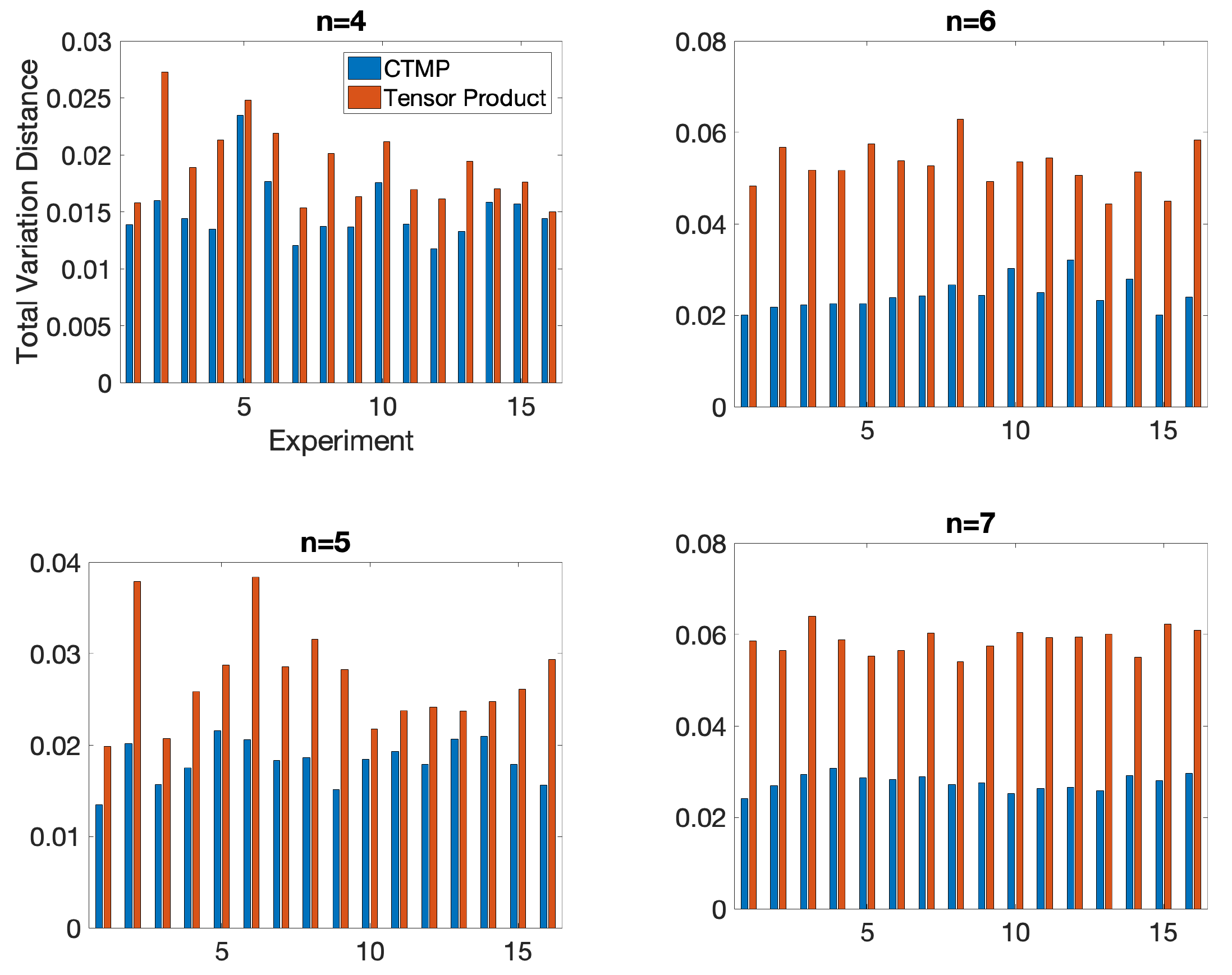}
	\caption{{\em Distance from the full $A$-matrix}. A comparison is made between the full
	$A$-matrix,
	the tensor product, and the CTMP noise models. The plot shows the
	total variation distance
	$\mathsf{TVD}(A_{\mathsf{full}},A_{\mathsf{ctmp}})$ (blue bars)
	and $\mathsf{TVD}(A_{\mathsf{full}},A_{\mathsf{tp}})$ (red bars) for $n=4,5,6,7$ qubits. 
	For each number of qubits we performed 
	$16$  independent experiments.
 \label{fig:TVDplot}}
\end{figure}

In order to compare the methods of measurement mitigation on relevant quantum circuits, we use graph states as an example of a highly entangled state that can serve as a benchmark for quantum devices.
We estimate the fidelity with stabilizer measurements for $n$-qubit graph states~\cite{raussendorf2001one} 
of the form
\be
\label{graph_state}
|G_n\ra = \left( \prod_{j=1}^{n-1} \mathsf{CZ}_{j,j+1} \right) \mathsf{H}^{\otimes n} |0^n\ra.
\ee
The corresponding graph is  a path with $n$ vertices.
The state $|G_n\ra$  is a stabilizer state with the stabilizer group $\calS$
generated by Pauli operators
$S_1 = X_1 Z_2$, $S_n=Z_{n-1} X_n$, and $S_j = Z_{j-1} X_j Z_{j+1}$
for $1<j<n$. Let $\rho$ is an approximate version of $|G_n\ra$
prepared in the lab
by executing the quantum circuit Eq.~(\ref{graph_state}).
The fidelity $\la G_n|\rho|G_n\ra$
can be estimated by picking a random element of the stabilizer group $S\in \calS$
and measuring its mean value on the state $\rho$,
\be
\label{eq:graphfidelity}
F=\la G_n|\rho|G_n\ra = 2^{-n} \sum_{S\in \calS} \mathrm{Tr}(\rho S)
=\EE_{S} \mathrm{Tr}(\rho S).
\ee
A comparison between the error-mitigated fidelity $F$
obtained using the full $A$-matrix, TP, and CTMP methods
with weight-$1$ calibration
 is shown on Fig.~\ref{fig:fidvnumq}(a).
We also show raw values of $F$ obtained without error mitigation.
The full $A$-matrix  model was only used for $n\le 7$. 
The data suggests that CTMP provides much more accurate
estimates of the fidelity, compared with the TP model
which systematically over-estimates the fidelity.
The mean value of each stabilizer was estimated by measuring
$M=8192$ copies of the graph state $|G_n\ra$.
The entire experiment (calibration and mean value measurements)
was repeated $16$ times  in order to estimate error bars.

\begin{figure}[h]
	\raggedright
	(a)\\
	\vspace{-10pt}\raggedleft	\includegraphics[width=0.95\columnwidth]{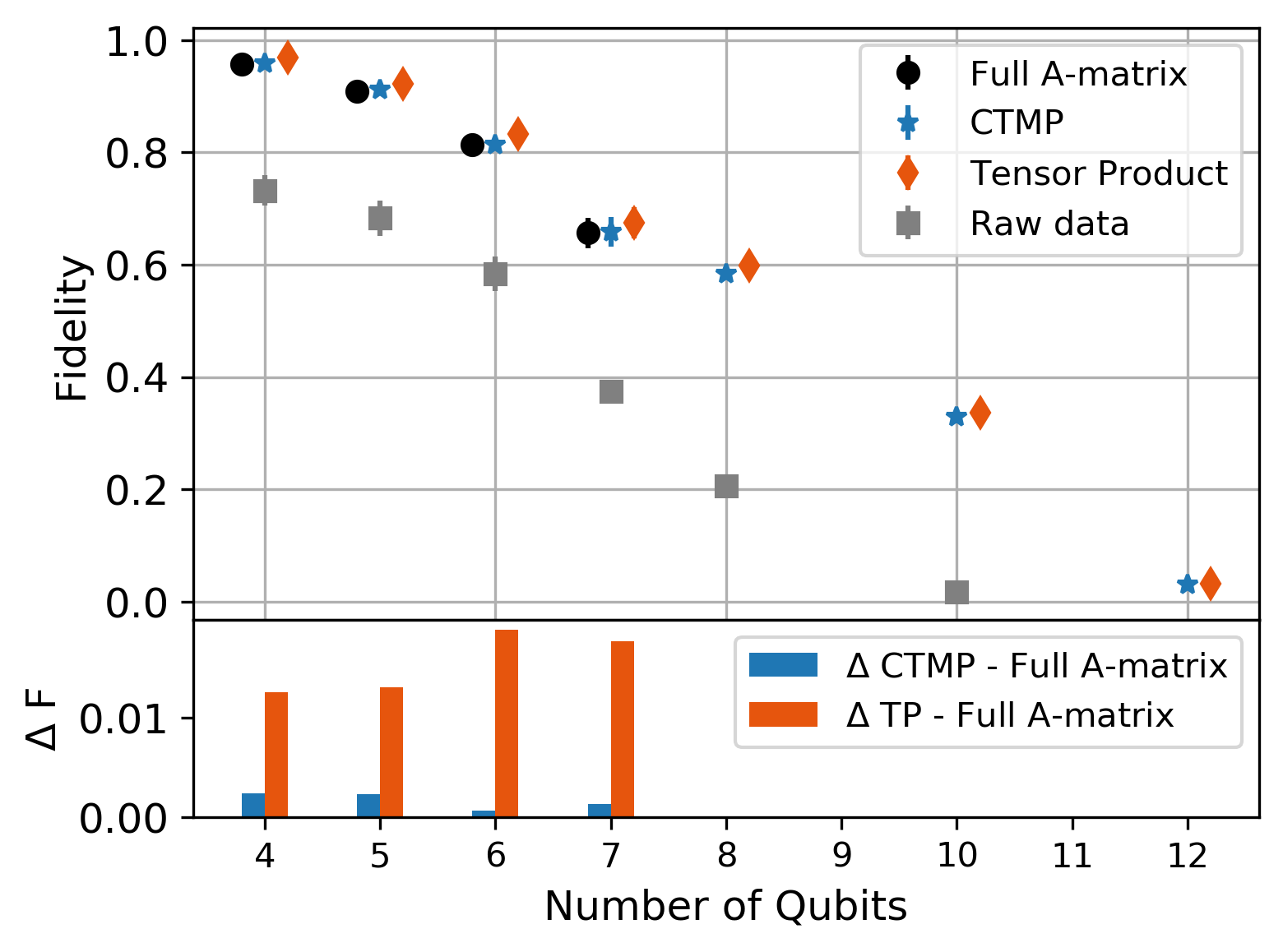}\\
	\raggedright
	(b)	\\
	\vspace{-10pt}\raggedleft
	\includegraphics[width=0.95\columnwidth]{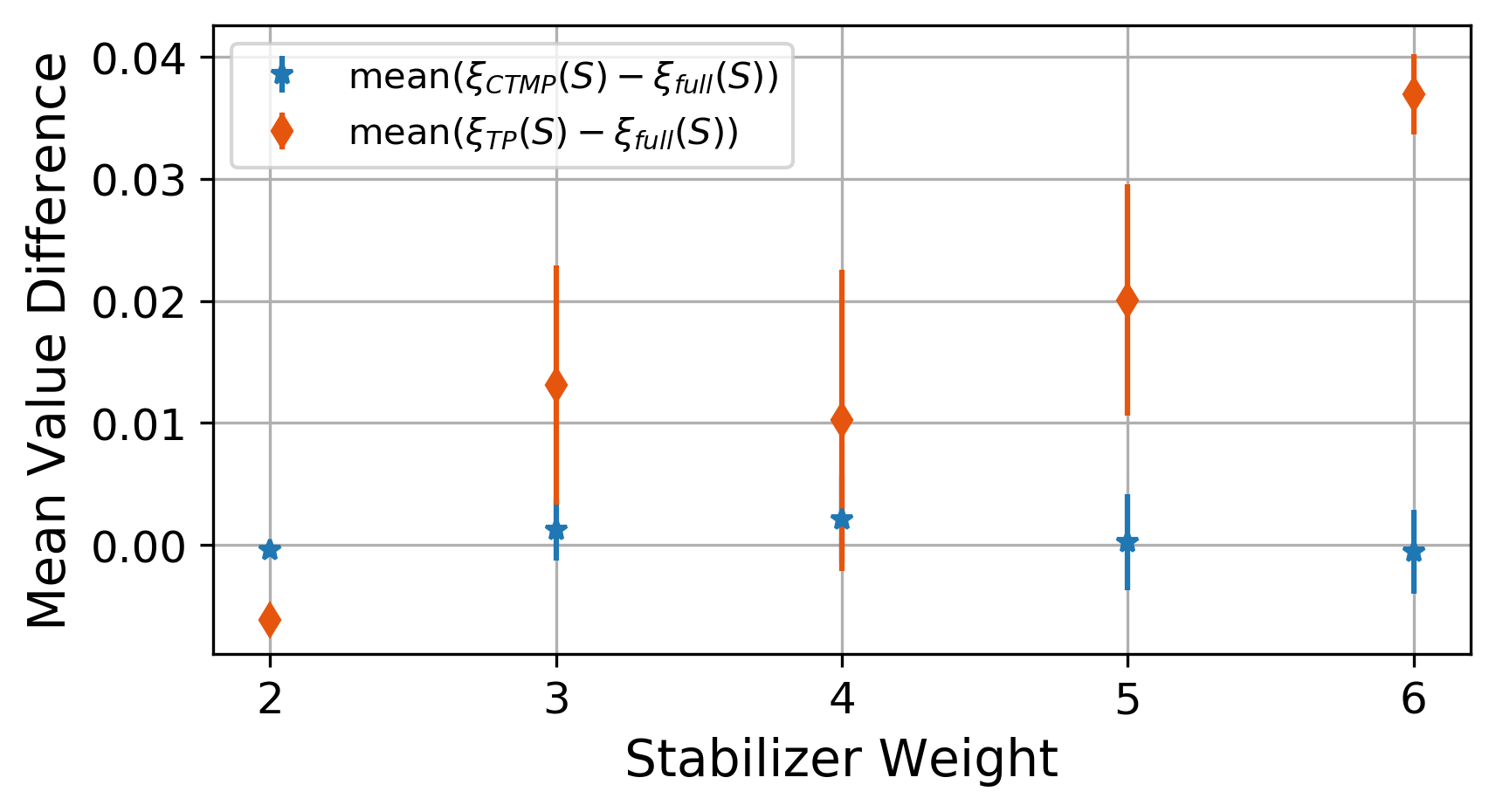}
	\caption{\label{fig:fidvnumq}
	(a) Average mean values of graph state stabilizers for $4\le n\le 12$ qubits
	obtained using   different error mitigation methods (Full A-matrix in black circles, CTMP in blue stars, TP in red diamonds) and without
	error mitigation (raw, grey squares) with the difference, $\Delta F$,  between CTMP and full A-matrix(blue) and TP and full A-matrix (red)  below. The data points on the main plot are offset for clarity.
	The fidelity was estimated by averaging the stabilizer mean value
	over $\approx 100$
	random stabilizers. For each stabilizer $16$ independent experiments were performed
	to estimate error bars.
   (b) Differences in the error-mitigated mean values $\xi_{\mathsf{ctmp}}(S)-\xi_{\mathsf{full}}(S)$ (blue) and
    $\xi_{\mathsf{tp}}(S)-\xi_{\mathsf{full}}(S)$ (red)
   for stabilizers $S$ of the $6$-qubit graph state from (a)
    obtained using the CTMP, the tensor product, and the full $A$-matrix
   noise models.
   Data sets are averaged over all stabilizers of the same weight, indicated on the 
   horizontal axis.
   The graph state has no stabilizers with weight=$1$,
   two stabilizers with weight=$2$, eight stabilizers with weight=$3$, etc. }
\end{figure}

For the six-qubit graph state, there are few enough stabilizers that we can measure all 
of them. Let $\xi_{\mathsf{full}}(S)$, $\xi_{\mathsf{tp}}(S)$, and $\xi_{\mathsf{ctmp}}(S)$ be 
the error-mitigated mean values $\mathrm{Tr}(\rho S)$ of a stabilizer $S$
obtained using the full $A$-matrix, the TP, and the CTMP method respectively.
Fig.~\ref{fig:fidvnumq}(b) shows the differences $\xi_{\mathsf{full}}(S)-\xi_{\mathsf{tp}}(S)$
and $\xi_{\mathsf{full}}(S)-\xi_{\mathsf{ctmp}}(S)$
averaged over stabilizers of the same weight, $|S|$. 
It can be seen that the TP method systematically over-estimates the mean values
 for stabilizers with the weight $3,4,5,6$,
 while CTMP shows a much better agreement with the full $A$-matrix method.  The difference between CTMP
 and TP increases for higher weight stabilizers, which we expect to be more sensitive to correlated readout errors.
Fig.~\ref{fig:cross_talk} (a) shows the combined rate of correlated two-qubit errors $01\leftrightarrow 10$ and
$00\leftrightarrow 11$, as described by the CTMP model, 
for each pair of qubits in the chosen 12-qubit register.
It can be seen that some pairs of qubits, such as $(7,8)$, experience 
correlated errors with a few-percent rate.

\begin{figure}[h]
	{(a) \includegraphics[height=5cm]{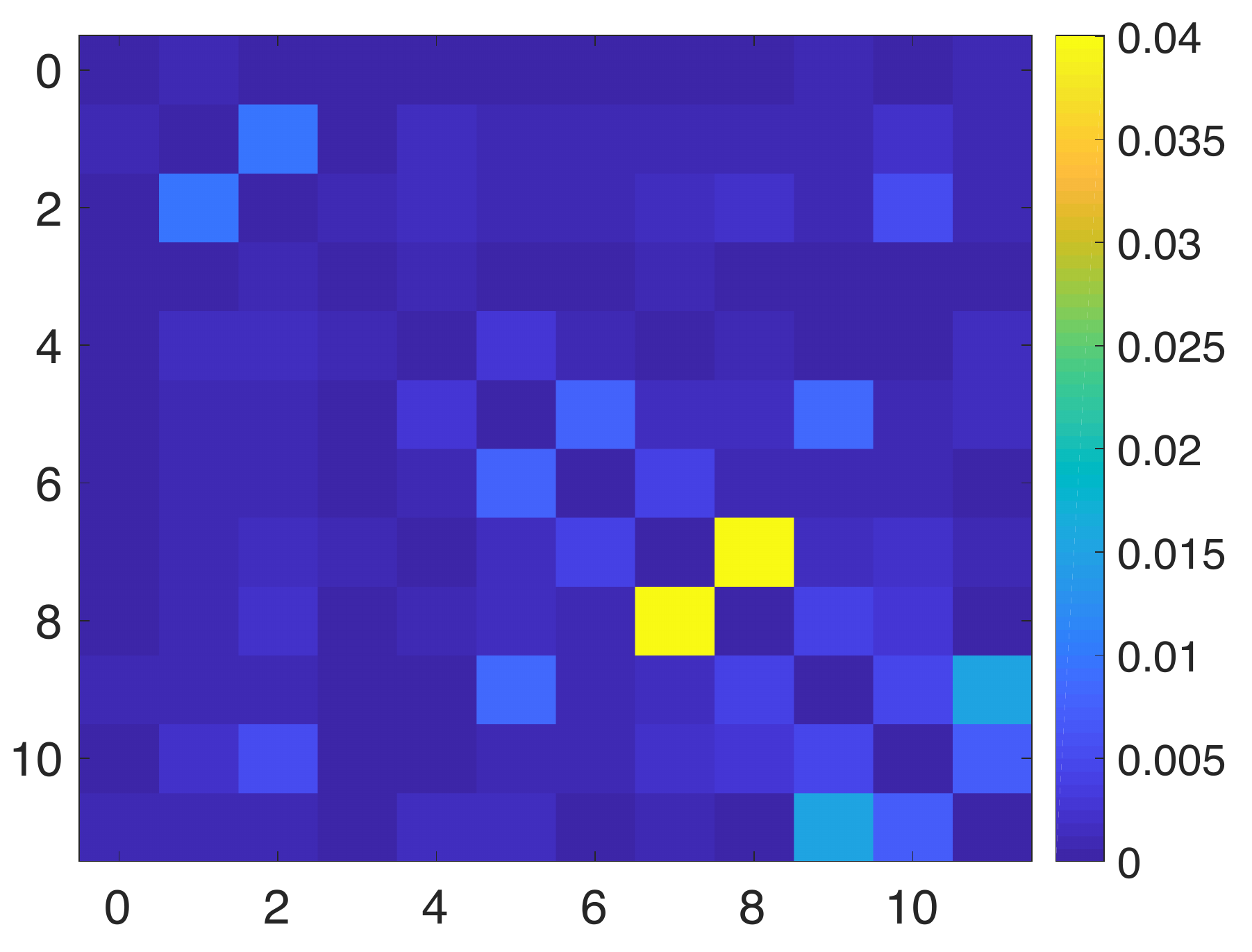}}\\
	{(b) \includegraphics[height=5cm]{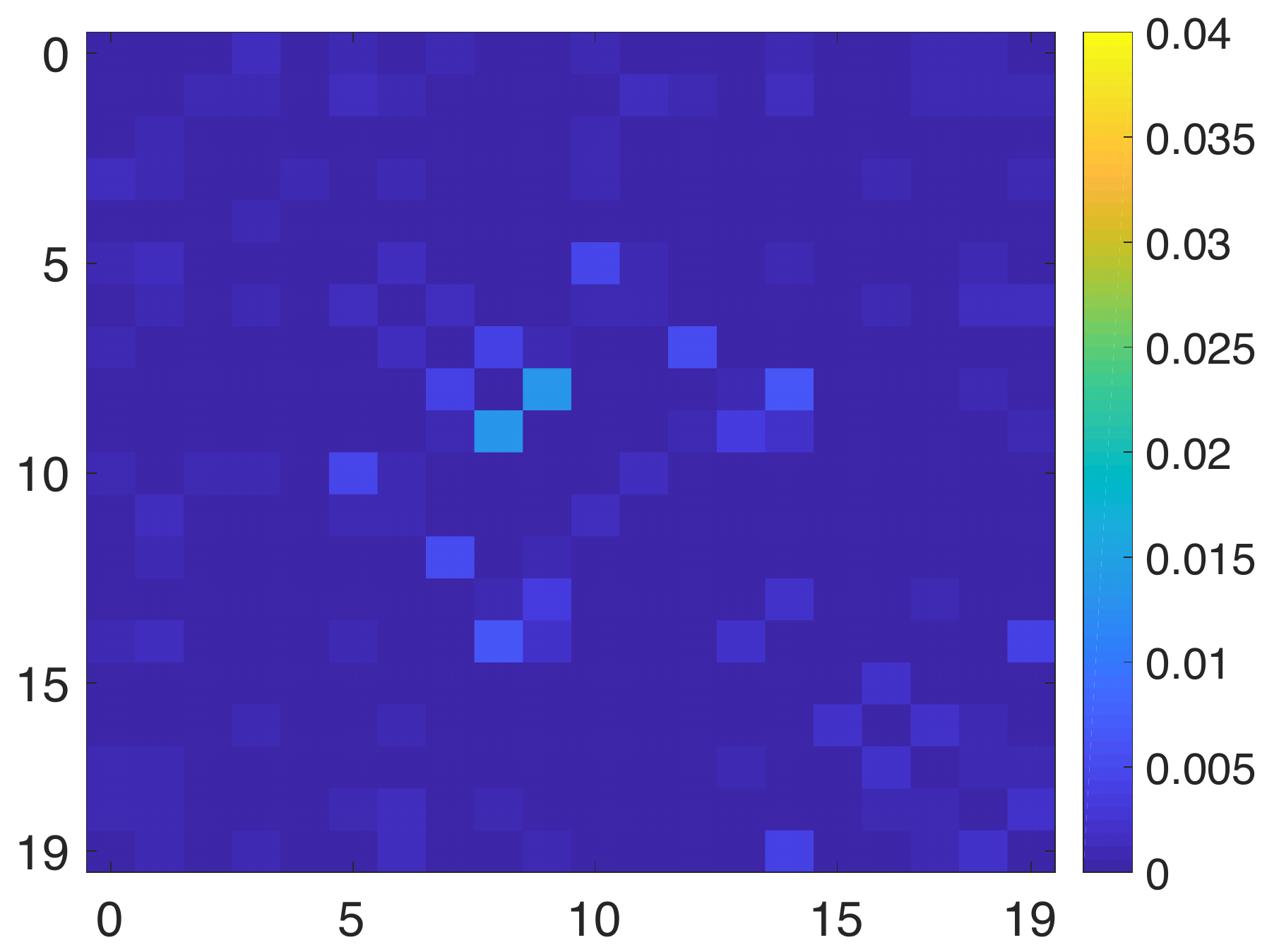}}
	\caption{Combined rate of correlated two-qubit readout errors  for each pair of qubits.(a) $12$-qubit graph state 
	experiment.(b) $20$-qubit random Clifford circuit experiment. 
 \label{fig:cross_talk}}
\end{figure}

Recall that the error mitigation overhead introduced by the CTMP method
is roughly $e^{4\gamma}$, where 
$\gamma$ is the noise strength defined in Eq.~(\ref{gammaG}).
Fig.~\ref{fig:gamma} shows the  noise strength $\gamma$ observed in the graph
state experiments
as a function of the number of qubits. Error bars were estimated from $16$
independent experiments. The data is consistent with a linear scaling,
$\gamma \approx 0.05 n$, and with single qubit readout fidelities of a few percent (see Appendix \ref{app:hardware} for more details on the readout errors in the hardware).
We observed $\gamma \le 0.7$ for all experiments reported above. 
The error mitigation overhead is therefore at most $e^{4\cdot 0.7}\le 20$, indicated a factor of $20$ increase in the number of measurements required.

\begin{figure}[h]
	\includegraphics[width=8cm]{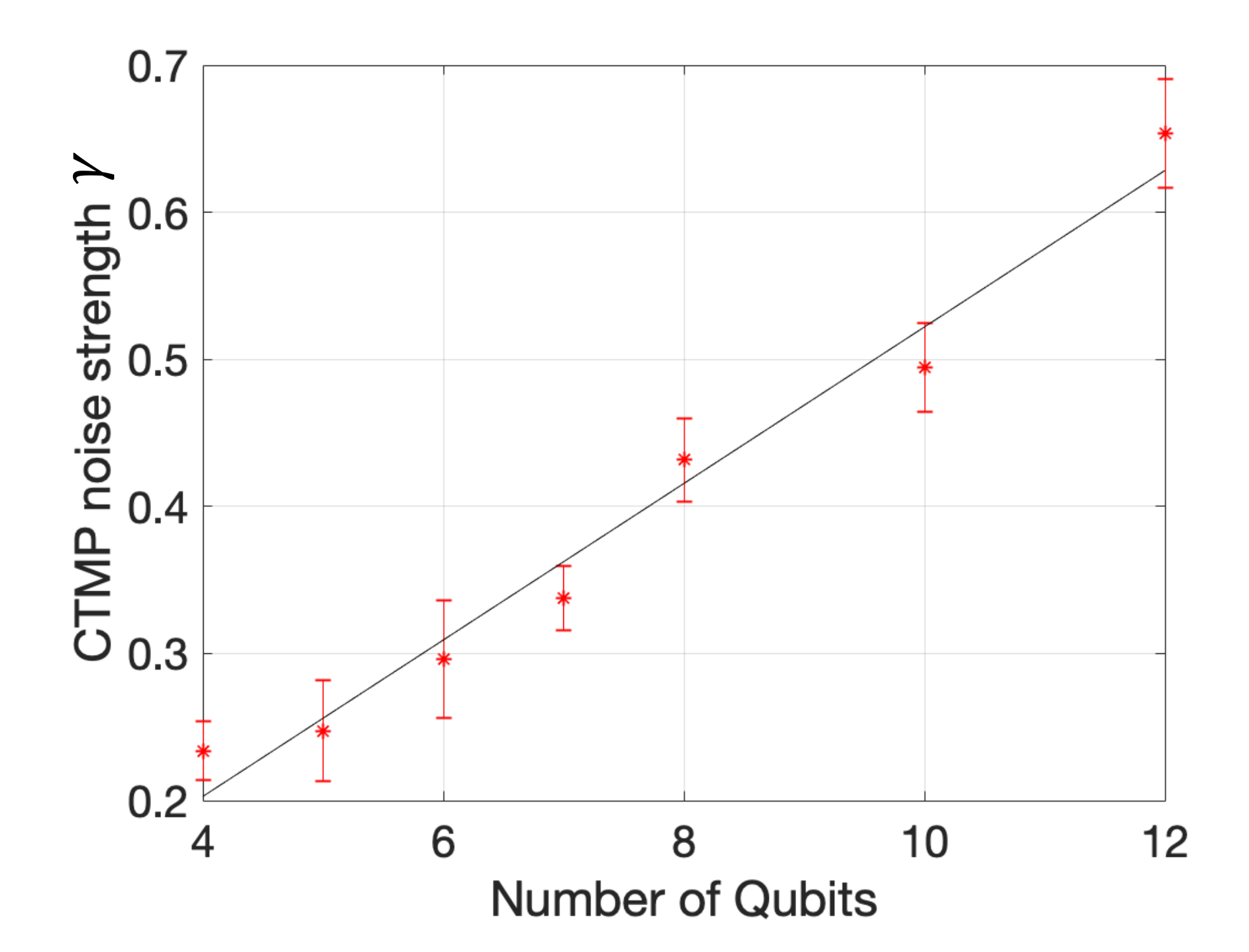}
	\caption{\label{fig:gamma} CTMP noise strength 
		$\gamma =\max_x -\la x|G|x\ra$, where
		$G$ is is generator matrix of the CTMP model such that $A_{\mathsf{ctmp}}=e^G$.
		The error mitigation overhead introduced by the CTMP method
		scales as $e^{4\gamma}$. The data suggests a linear scaling, $\gamma \approx 0.05 n$.
		We observed $\gamma \approx 1.1$ in the $n=20$ experiment.
	}
\end{figure}
\begin{figure*}
	{(a)\includegraphics[height=5.5cm]{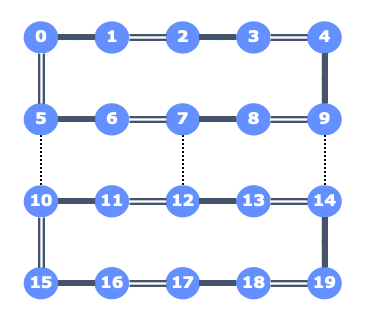} }
	{(b)\includegraphics[height=5.5cm]{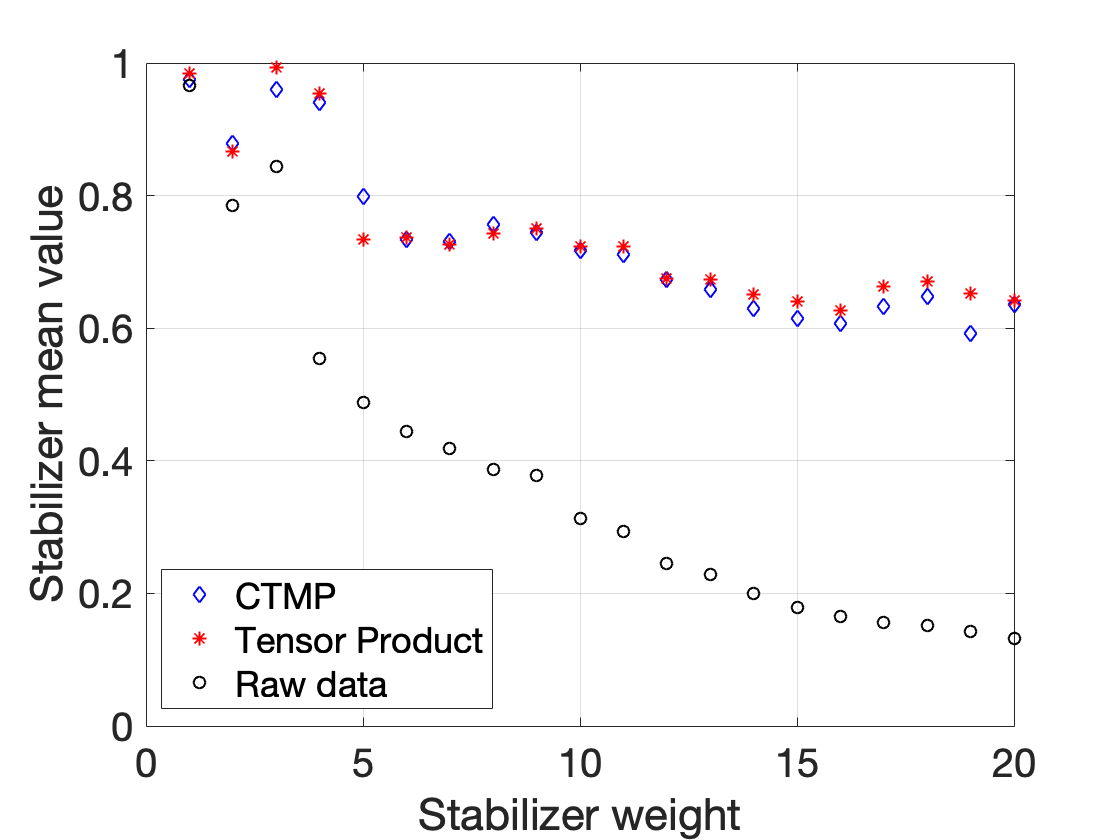} }
	\caption{\label{fig:rc_part2} Error mitigated stabilizer measurements
	performed on a $20$-qubit state $|\psi\ra$ generated by a random depth-4 Clifford circuit.
	The circuit consists of alternating layers of single-qubit Clifford operators 
	and entangling layers composed of nearest-neighbor CNOT gates. (a) the coupling map for the 20-qubit   device, $\mathsf{ibmq\_johannesburg}$.
	The patterns of CNOTs 
	in the two entangling layers are indicated by single and double solid lines, with unused physical connections in dotted lines.
	(b) Error mitigated and raw mean values $\la \psi|S|\psi\ra$
	for $\approx 500$ stabilizer operators $S$ chosen such that the weight of $S$ varies between $1$ and $20$.
	Each data point represents the average over all stabilizers $S$ with a given weight.
	In the absence of gate and readout errors each stabilizer $S$ has mean value $+1$.
	Error mitigation was performed using the CTMP and Tensor Product  methods with the Hadamard-type
	calibration.
	Black circles indicate raw mean values measured without error mitigation.
	}
\end{figure*}
 
 Our last experiment demonstrates error mitigated
 stabilizer measurements on an entangled  $20$-qubit  state $|\psi\ra$
 generated by a random depth-$4$ Clifford circuit.
The circuit consists of two layers of CNOT gates on nearest-neighbor pairs of qubits and two layers of single-qubit Clifford operators, see Fig.~\ref{fig:rc_part2}.
Such state $|\psi\ra$ can be specified by a group of Pauli-type stabilizers $S\in \pm \{I,X,Y,Z\}^{\otimes 20}$
such that $S|\psi\ra = |\psi\ra$.
We have measured mean values $\la \psi|S|\psi\ra$
for $\approx 500$ stabilizers $S$ chosen randomly such that the weight of $S$ varies between $1$ and $20$.
In the absence of readout and gate errors one has $\la \psi|S|\psi\ra=1$ for all stabilizers $S$.
 Our results for the error mitigated mean values $\la \psi|S|\psi\ra$
obtained with the TP and CTMP methods and Hadamard-type calibration
are shown on Fig.~\ref{fig:rc_part2}. 
In contrast to the graph state experiment, error rates associated with
correlated two-qubit errors $01\leftrightarrow 10$ and $00\leftrightarrow 11$
observed in the $20$-qubit experiment are less than $1\%$ for all pairs of qubits,
see Fig.~\ref{fig:cross_talk}(b).  Accordingly, 
the difference between error mitigated mean values obtained using
the TP and CTMP methods is less pronounced.

The $20$-qubit experiment showcases the scalability of our method.
For example, 
computing the error-mitigated mean value of a $20$-qubit  stabilizer
using the CTMP method with $T=10^6$ samples in Algorithm~1 takes about $2$ seconds
on a laptop computer (recall that $T$ affects the cost of a classical post-processing but not the
number of experiments).
Extracting CTMP error rates from the $20$-qubit calibration data
takes about $1$ second.
Thus we expect that our methods can be scaled up to a larger number of qubits.

\section{conclusions}

We  introduced scalable methods of mitigating readout errors in multi-qubit experiments.
Our methods are capable of mitigating correlated cross-talk errors and enable
efficient implementation of the calibration and the noise inversion steps.
The proposed error mitigation methods are demonstrated experimentally for
measurements of up to $20$ qubits. We believe that our techniques will be useful for many near term quantum applications, and their scalability will be increasingly important as the problem sizes become larger.

\begin{acknowledgements}
The authors thank Kristan Temme for helpful discussions and Neereja Sundaresan for experimental contributions.
This work was supported in part by ARO under Contract No. W911NF-14-1-0124
and by the IBM Research Frontiers Institute.
The authors declare that they have no competing interests.
\end{acknowledgements}

\appendix

\section{Proof of Lemma~\ref{lemma:xi}}
\label{app:1}

Represent the measurement outcomes $s^1,\ldots,s^M$ by a 
probability vector $|\psi\ra = M^{-1}\sum_{i=1}^M |s^i\ra$.
Then
\[
\xi = \sum_x O(x) \la x|A^{-1} |\psi\ra
\]
and
\be
\label{var1}
\EE(\xi^2)=\sum_{x,y} O(x) O(y) \la x|A^{-1} \EE(|\psi\ra\la \psi|) (A^{-1})^T |y\ra.
\ee
Let $P(x)=\la x|\rho|x\ra$ and $|P\ra = \sum_x P(x)|x\ra$.
Then 
\be
\label{var2}
\EE(|\psi\ra\la \psi|) = \frac{M-1}M A|P\ra\la P|A^T + \frac1M \sum_x |x\ra\la x|  \la x|A|P\ra.
\ee
Substituting Eq.~(\ref{var2}) into Eq.~(\ref{var1}) gives
\be
\EE(\xi^2)\le \EE(\xi)^2  + \frac1M \sum_{x,y} O(x) O(y) \la x|A^{-1} D (A^{-1})^T |y\ra,
\ee
where $D$ is a diagonal operator  with entries $\la x|D|x\ra = \la x|A|P\ra$.
Since $\la x|A|P\ra$ is a normalized probability distribution, one 
concludes that $\EE(\xi^2) - \EE(\xi)^2$ is upper bounded by 
\be
\frac1M \max_z \sum_{x,y} |O(x)O(y)| \cdot |\la x|A^{-1} |z\ra|\cdot |\la y|A^{-1}|z\ra|.
\ee
By assumption,  $|O(x)|\le 1$ for all $x$, which gives
$\EE(\xi^2) - \EE(\xi)^2\le M^{-1}\Gamma^2$.
Thus the standard deviation of $\xi$ is at most $M^{-1/2}\Gamma$.

\section{Proof of Lemma~\ref{lemma:QPR}}
\label{app:stochastic}

Suppose $M$ is a real $N\times N$ matrix
(we shall be interested in the case $M=A^{-1}$).
Let us  state necessary and sufficient conditions
under which $M$ admits a decomposition
\be
\label{SD}
M=\sum_a c_a S_a
\ee
where $S_a$ are stochastic matrices
and $c_a$ are real coefficients.
Define the $j$-th column sum of $M$ as
$\sigma_j(M)=\sum_{i=1}^N M_{i,j}$.
\begin{lemma}
	A  matrix $M$ admits a decomposition Eq.~(\ref{SD}) iff
	all column sums of $M$ are the same. 
\end{lemma}
\begin{proof}
	From Eq.~(\ref{SD}) one gets
	$\sigma_j(M)=\sum_a c_a$ for all $j$. 
	
	Conversely, suppose $\sigma_j(M)=\gamma$ for all $j$.
	For any integers $i,j\in [N]$ define
	a matrix $T(i,j)=|1\ra \la j| - |i\ra \la j|$.
	Note that all column sums of $T(i,j)$ are zero.
	One can easily verify that $T(i,j)$ is a linear combination
	of two stochastic matrices.
	Let
	\[
	M'=M + \sum_{i=2}^N \sum_{j=1}^N M_{i,j} T(i,j).
	\]
	The matrix $M'$ has zero rows $2,3,\ldots,N$
	and $M'_{1,j}=\gamma$ for all $j$.
	Thus $M'=\gamma S$ for a stochastic matrix $S$.
	This shows that $M$  is a linear combination of stochastic matrices.
\end{proof}
As a corollary, the inverse of any stochastic matrix can be written
as a linear combination of stochastic matrices. Indeed, if 
$M=A^{-1}$ for some stochastic matrix $A$
then all column sums of $M$ are equal to one.

Given a matrix $M$ let $\Gamma(M)$ be the maximum $1$-norm
of its columns,
\[
\Gamma(M)= \max_{j} \; \sum_{i=1}^N |M_{i,j}|.
\]
Lemma~\ref{lemma:QPR} is a special case of the following.
\begin{lemma}
	Suppose all column sums of $M$ are the same. Then
	any decomposition Eq.~(\ref{SD})
	obeys $\|c\|_1 \ge \Gamma(M)$ and
	$\|c\|_1=\Gamma(M)$ for some decomposition.
\end{lemma}
\begin{proof}
	Suppose $M$ is written as in Eq.~(\ref{SD}).
	By triangle
	inequality,
	\be
	\label{SD1}
	\Gamma(M)= \sum_{i} |M_{i,j}| \le \sum_a |c_a| \sum_i \la i|S_a|j\ra = \|c\|_1
	\ee
	Conversely, let $\Lambda(M)$ be the minimum $1$-norm $\|c\|_1$,
	where the minimum is over all decompositions Eq.~(\ref{SD}).
	From Eq.~(\ref{SD1}) one infers that $\Lambda(M)\ge \Gamma(M)$.
	It suffices to prove that
	\be
	\label{SD2}
	\Lambda(M)\le \Gamma(M).
	\ee
	Let $k$ be the number of non-zero elements in $M$.
	We shall prove Eq.~(\ref{SD2}) using induction in $k$. The base of induction is $k=0$.
	In this case $M$ is the all-zeros matrix and $\Lambda(M)=\Gamma(M)=0$.
	
	Suppose we have already proved  Eq.~(\ref{SD2}) for all
	matrices $M$ with at most $k$ non-zeros. 
	Consider a matrix $M$ with $k+1$ non-zeros.
	Let $\gamma=\sigma_j(M)$ be the column sum of $M$
	(by assumption, it does not depend on $j$).
	We shall consider two cases.\\
	\noindent
	{\em Case~1:} $\gamma\ne 0$.
	Then all columns of $M$ are non-zero and each column contains
	a non-zero entry $M_{i,j}$  with the same sign as $\gamma$.
	Let $\omega$ be a smallest magnitude non-zero entry of $M$
	with the same sign as $\gamma$.
	Assume wlog that $\omega=M_{1,1}$ (otherwise, permute rows and/or columns of $M$).
	Choose any function $f\, : \, [N]\to [N]$ such that 
	$f(1)=1$ and $M_{f(j),j}$ is a largest magnitude entry in the $j$-th
	column with the same sign as $\gamma$.
	Define a new matrix
	\be
	\label{SD3}
	M'=M - \omega \sum_{j=1}^N |f(j)\ra\la j|.
	\ee
	Note that $M'$ has at most $k$ non-zeros
	since $M_{f(j),j}\ne 0$ for all $j$ and $M'_{1,1}=0$.
	Since $\sum_{j=1}^N |f(j)\ra\la j|$ is a stochastic matrix, Eq.~(\ref{SD3}) gives 
	\be
	\label{SD4}
	\Lambda(M)\le \Lambda(M') + |\omega|.
	\ee
	We claim that 
	\be
	\label{SD5}
	\Gamma(M') \le \Gamma(M)-|\omega|.
	\ee
	Indeed, let $\Gamma_j(M)$ and $\Gamma_j(M')$ be the $1$-norm of the $j$-th column
	of $M$ and $M'$ respectively. We have $\Gamma_1(M')=\Gamma_1(M)-|\omega|$
	since the first column of $M'$ is obtained from the one of $M$ by setting to zero
	the entry $M_{1,1}=\omega$. If $j\ge 2$ then
	\begin{eqnarray}
	\label{SD6}
	\Gamma_j(M') &=& |M_{f(j),j}-\omega| +  \sum_{i\ne f(j)} |M_{i,j}| \nonumber \\
	&=&  \mathrm{sgn}(\omega)\cdot (M_{f(j),j}-\omega) + \sum_{i\ne f(j)} |M_{i,j}|  \nonumber\\ 
	&=&\Gamma_j(M) -|\omega|. \nonumber
	\end{eqnarray}
	Here we used $\mathrm{sgn}(\omega)=\mathrm{sgn}(M_{f(j),j})$ 
	and $|\omega|\le |M_{f(j),j}|$.
	This proves Eq.~(\ref{SD5}).
	By induction hypothesis, $\Lambda(M')\le \Gamma(M')$. Combining this
	and Eqs.~(\ref{SD4},\ref{SD5}) one arrives at
	\be
	\Lambda(M)\le \Lambda(M')+|\omega| \le \Gamma(M')+|\omega| \le   \Gamma(M).
	\ee
	
	\noindent
	{\em Case~2:} $\gamma=0$. Then each non-zero column of $M$ contains  at least
	one positive and at least one negative entry.  
	Let $\omega$ be a non-zero entry of $M$ with the smallest magnitude.
	Assume wlog that $\omega=M_{1,1}$.
	Choose any function $f\, : \, [N]\to [N]$ such that 
	$f(1)=1$ and  $M_{f(j),j}$ is the largest magnitude entry in the $j$-th
	column with the same sign as $\omega$.
	If the $j$-th column is zero then set $f(j)$ arbitrarily. 
	Define a new matrix $M'$ using Eq.~(\ref{SD3}). The same arguments as above
	show that $\Gamma_1(M')=\Gamma_1(M)-|\omega|$ and
	$\Gamma_j(M')\le \Gamma_j(M)-|\omega|$ whenever the $j$-th column of $M$
	is non-zero. Suppose now that the $j$-th column of $M$ is zero.
	Then $\Gamma_j(M)=0$ and $\Gamma_j(M')=|\omega|$.
	We claim that $\Gamma(M)\ge 2|\omega|$. Indeed, by assumption, $M$
	has at least one non-zero column. Such column must have at least two non-zero
	entries with the magnitude at least $|\omega|$, that is,
	$\Gamma(M)\ge 2|\omega|$. Thus, if the $j$-th column of $M$ is zero then
	$\Gamma_j(M')=|\omega|\le \Gamma(M)-|\omega|$.
	This proves Eq.~(\ref{SD5}).The same arguments as above give $\Lambda(M)\le \Gamma(M)$.
\end{proof}

\section{Error mitigation for product noise}
\label{app:product}

Here we  instantiate Algorithm~1 for 
the tensor product noise defined in 
Eq.~(\ref{product_noise}).
For each $j\in [n]$ 
define the following quantities
\[
c_{j0} = \frac{2-\epsilon_j-\eta_j}{2(1-\epsilon_j-\eta_j)},
\qquad
c_{j1} = -\frac{\epsilon_j+\eta_j}{2(1-\epsilon_j-\eta_j)},
\]
\[
c_{j2} = \frac{\epsilon_j-\eta_j}{2(1-\epsilon_j-\eta_j)},
\qquad
c_{j3}= \frac{\eta_j-\epsilon_j}{2(1-\epsilon_j-\eta_j)}.
\]
Define single-bit boolean functions $F_0,F_1,F_2,F_3$ as
\[
F_0(b)=b, \quad F_1(b)=1\oplus b,\quad F_2(b)=0, \quad F_3(b)=1.
\]
Here $b\in \{0,1\}$. Consider the following
algorithm.
\begin{center}
	\fbox{\parbox{\linewidth}{
			{\bf Algorithm 2}
			\begin{algorithmic}
				\State{$T\gets  4\delta^{-2}\Gamma^2 $}
				\For{$t=1$ to $T$}
				\State{$\xi^t \gets 1$}
				\State{Sample $i\in [M]$ uniformly at random.}
				\State{$x^t \gets s^i$}
				\For{$j=1$ to $n$}
				\State{\parbox{6.6cm}{Sample $\alpha \in \{0,1,2,3\}$ with probabilities $q_{j\alpha}=
						|c_{j\alpha}|(|c_{j0}|+|c_{j1}| + |c_{j2}| + |c_{j3}|)^{-1}$}}
				\State{$x^t_j \gets F_\alpha (x^t_j)$}
				\State{$\xi^t \gets \xi^t \cdot \mathrm{sgn}(c_{j\alpha})$}
				\EndFor
				\EndFor
				\State{\Return{$\xi'=T^{-1} \Gamma \sum_{t=1}^T \xi_t O(x^t)$}}
			\end{algorithmic}
	}}
\end{center}
Here $x^t_j\in \{0,1\}$ is the $j$-th bit of $x^t$.
We claim that  Algorithm~2 is a special case of Algorithm~1 stated
in the main text.
Indeed, define stochastic matrices
\[
A_j= \left[ \ba{cc} 1-\epsilon_j & \eta_j \\
\epsilon_j & 1-\eta_j \\ \ea \right]
\quad \mbox{and} \quad
V_i = \sum_{b=0,1} |F_i(b)\ra\la b|.
\]
A simple calculation shows that 
\be
\label{quasiProb1}
A_j^{-1} 
=c_{j0} V_0 + c_{j1} V_1 + c_{j2} V_2  + c_{j3} V_3.
\ee
and
\be
\label{quasiProb2}
|c_{j0}|+|c_{j1}| + |c_{j2}| + |c_{j3}| = \frac{1+|\epsilon_j-\eta_j|}{1-\epsilon_j-\eta_j}.
\ee
Taking $n$-fold tensor product of stochastic decompositions
defined
in Eq.~(\ref{quasiProb1})
one gets $A^{-1}=\sum_\alpha c_\alpha S_\alpha$, where
$S_\alpha$ is a tensor product of stochastic matrices $V_i$
and $c_\alpha$ are real coefficients such that $\|c\|_1=\Gamma$.
The inner loop of Algorithm~2  samples $x^t$ from the distribution
$S_\alpha |s^i\ra$,
where $\alpha$ is sampled from $q_\alpha=|c_\alpha|/\Gamma$.
 Thus Algorithm~2 is indeed
a special case of Algorithm~1.
We conclude that the output of 
Algorithm~2  satisfies
$|\xi' - \xi|\le \delta$ with high probability
(at least $2/3$).
The algorithm has runtime $\approx nT\tau \sim nM \tau$, where
$M$ is the number of measurements determined by Eqs.~(\ref{shots_upper_bound},\ref{GammaProduct})
and $\tau$ is the cost of computing the observable $O(x)$ for a given $x$.
For simplicity, here we assume that each sampling step in Algorithm~2 can be done in unit time.

\section{Experimental Hardware}
\label{app:hardware}

\begin{figure*}[h]
	(a)\includegraphics[width=0.9\columnwidth]{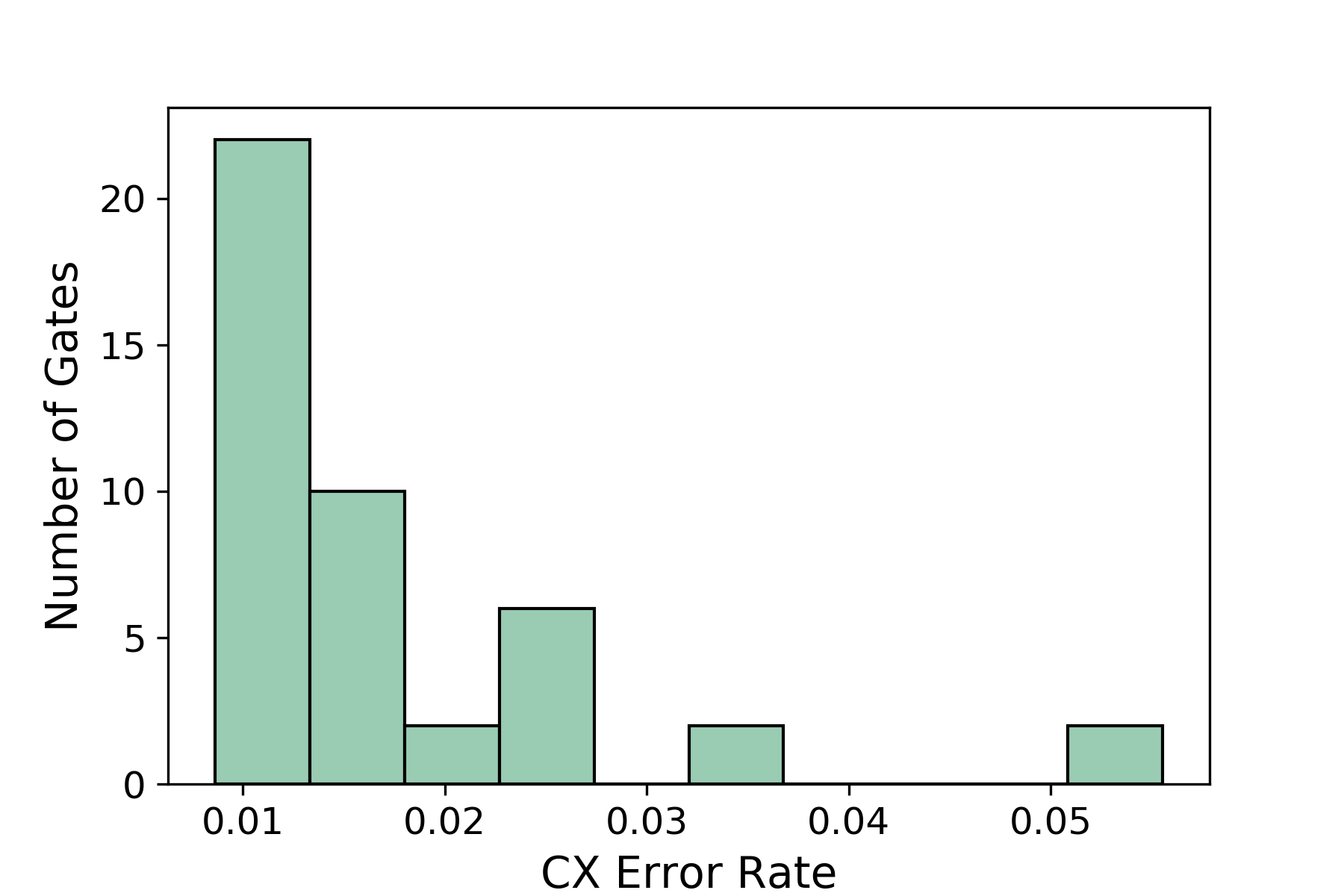}
	(b)\includegraphics[width=0.9\columnwidth]{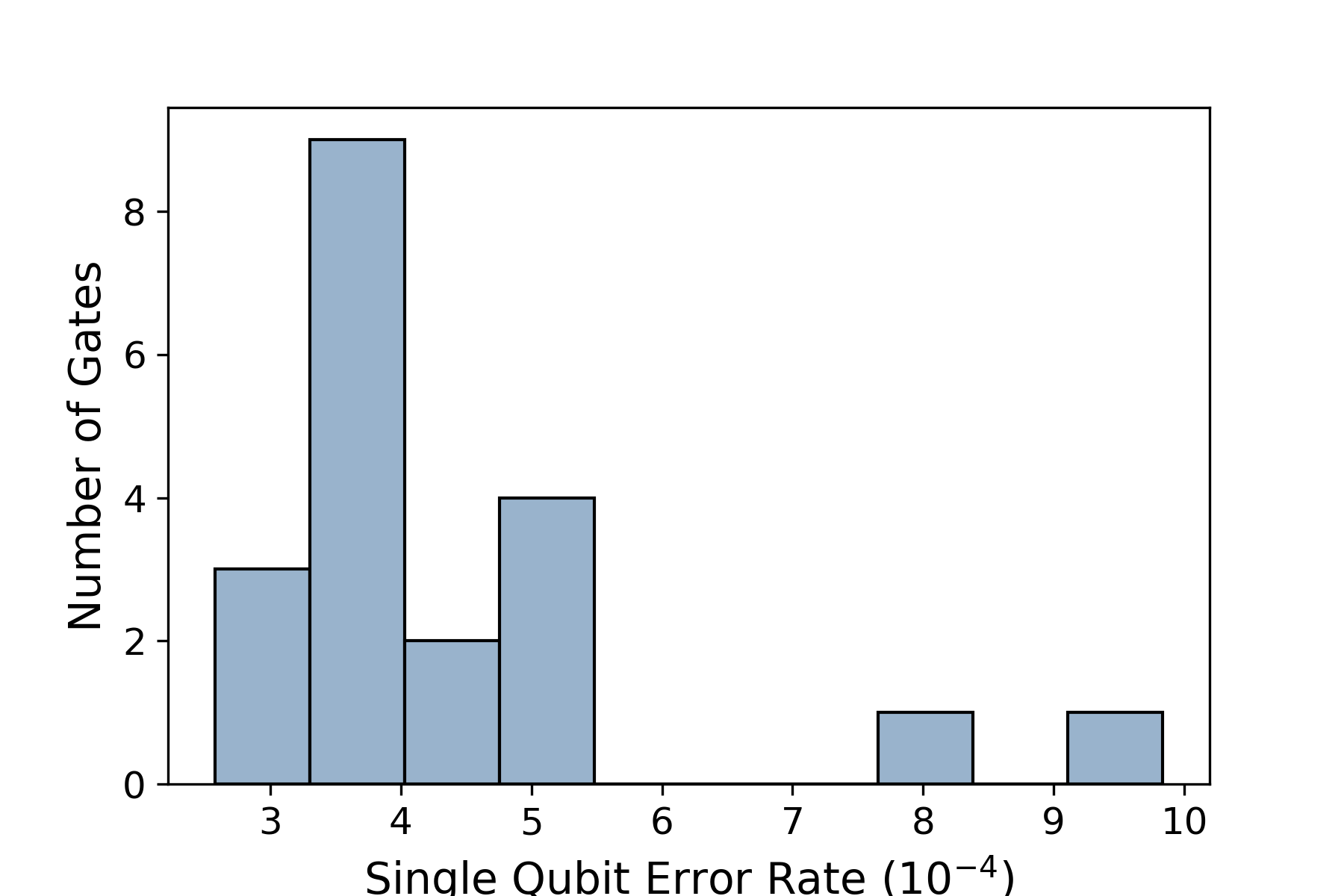}
	(c)\includegraphics[width=0.9\columnwidth]{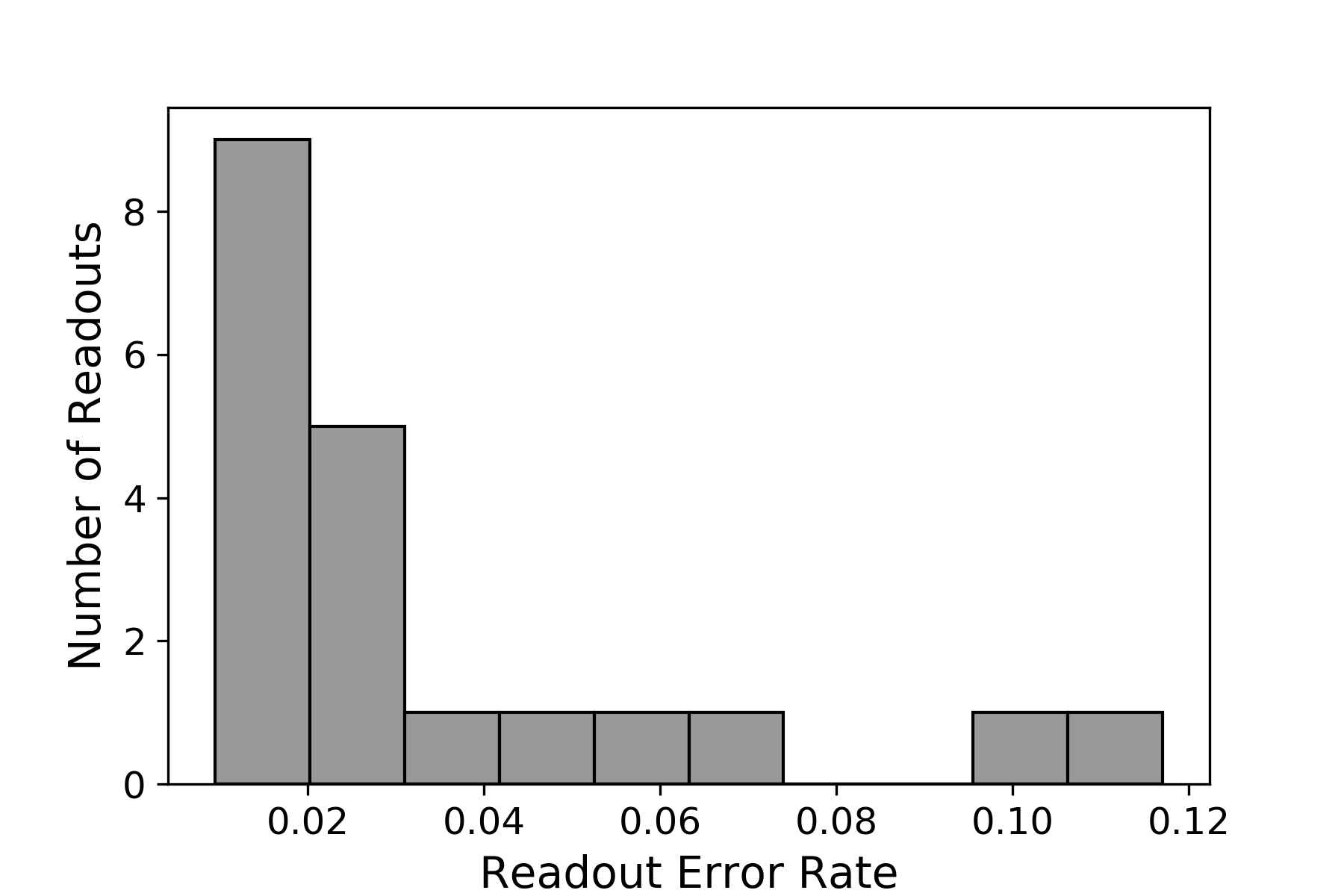}
	\caption{\label{fig:errorrates} Histograms of the (a) two-qubit CX (b) single-qubit and (c) readout error rates across the full 20-qubit device, $\mathsf{ibmq\_johannesburg}$, for a typical set of calibrations.
	}
\end{figure*}

All experiments were performed on the $20$-qubit IBM Quantum device
called $\mathsf{ibmq\_johannesburg}$.  The device coupling map is shown in Fig.~\ref{fig:rc_part2}(a) of the main text.  Typical mean error  rates are $(1.75 \pm 1.05) \times 10^{-2}$ for CX gates, $(4.35\pm1.70) \times 10^{-4}$ for single qubit gates, and $(3.44\pm1.72) \times 10^{-2}$ for readout errors; a more detailed breakdown of the error rates is given in the histograms in Fig.~\ref{fig:errorrates}.  
The graph state experiments in Fig.~\ref{fig:fidvnumq}
were performed on $n$-qubit subsets of qubits 0 through 11.

\end{document}